\documentclass[reqno]{article}
\usepackage{amssymb}
\usepackage{amsmath}
\usepackage{amsthm}
\usepackage{graphicx,subfigure,mathrsfs}
\usepackage[usenames]{color}

\newcommand{\threevec}[3]
{ \left( \begin{array}{c} #1 \\ #2 \\#3 \end{array} \right) }

\newcommand{\odone}[2]{\frac{d #1}{d #2}}

\newcommand{\pdone}[2]{\frac{\partial #1}{\partial #2}}

\newcommand{\la}{\langle}
\newcommand{\ra}{\rangle}

\newcommand{\ep}{\varepsilon}
\newcommand{\os}{\omega_0^2}

\newcommand{\Id}{ {\cal I} }

\newcommand{\R}{\mathbb{R}}
\newcommand{\C}{\mathbb{C}}

\newcommand{\D}{\mathcal{D}}
\newcommand{\bD}{{\bf D}}
\newcommand{\bA}{{\bf A}}

\newcommand{\bS}{{\bf S}}

\newcommand{\ba}{{\bf a}}
\newcommand{\bB}{{\bf B}}
\newcommand{\bH}{{\bf H}}
\newcommand{\bs}{{\bf s}}

\newcommand{\bfphi}{\boldsymbol \phi}
\newcommand{\bfpsi}{\boldsymbol \psi}

\newcommand{\bb}{{\bf b}}
\newcommand{\bm}{{\bf m}}
\newcommand{\bM}{{\bf M}}

\newcommand{\bx}{{\bf x}}

\newcommand{\bc}{{\bf c}}

\newcommand{\bfa}{ {\bf a}}
\newcommand{\bfb}{ {\bf b}}
\newcommand{\bfc}{ {\bf c}}

\newcommand{\bfg}{ {\bf g}}
\newcommand{\bfh}{ {\bf h}}

\newcommand{\bfm}{ {\bf m}}

\newcommand{\bfp}{ {\bf p}}
\newcommand{\bfq}{ {\bf q}}

\newcommand{\bfs}{ {\bf s}}

\newcommand{\bfu}{ {\bf u}}
\newcommand{\bfv}{ {\bf v}}
\newcommand{\bfw}{ {\bf w}}
\newcommand{\bfx}{ {\bf x}}
\newcommand{\bfy}{ {\bf y}}
\newcommand{\bfz}{ {\bf z}}
\newcommand{\bfca}{ {\bf A}}
\newcommand{\bfcb}{ {\bf B}}

\newcommand{\bfcd}{ {\bf D}}
\newcommand{\bfce}{ {\bf E}}
\newcommand{\bfcf}{ {\bf F}}
\newcommand{\bfcg}{ {\bf G}}
\newcommand{\bfch}{ {\bf H}}
\newcommand{\bfci}{ {\bf I}}

\newcommand{\bfcl}{ {\bf L}}
\newcommand{\bfcm}{ {\bf M}}

\newcommand{\bfcp}{ {\bf P}}
\newcommand{\bfcq}{ {\bf Q}}
\newcommand{\bfcr}{ {\bf R}}
\newcommand{\bfcs}{ {\bf S}}
\newcommand{\bfct}{ {\bf T}}
\newcommand{\bfcu}{ {\bf U}}
\newcommand{\bfcv}{ {\bf V}}
\newcommand{\bfcw}{ {\bf W}}

\newcommand{\bfcy}{ {\bf Y}}

\newcommand{\bfS}{{\bf \Sigma}}

\newcommand{\mcl}{\mathcal{L}}

\newcommand{\omu}{\overline{\mu}}

\newcommand{\bI}{{\bf I}}
\newcommand{\bZero}{{\bf 0}}

\newcommand{\ip}[2]{\langle #1 , #2 \rangle}
\newcommand{\avg}[1]{\langle \! \langle #1\rangle \! \rangle}

\newcommand{\og}{4\os - \gamma^2}

\newcommand{\gzero}{{\bf \Gamma}_0}
\newcommand{\gone}{{\bf \Gamma}_1}
\newcommand{\azero}{{\bf A}_0}
\newcommand{\aone}{{\bf A}_1}

\newcommand{\as}{\langle \bfa , \bfs \rangle}
\newcommand{\ats}{\langle \bfa , \bfs (t)\rangle}
\newcommand{\rp}[1]{\mbox{Re}\left[ #1 \right]}
\newcommand{\tr}[1]{\mbox{tr}\left( #1 \right)}
\newcommand{\eq}[1]{(\ref{eq:#1})}
\newcommand{\ob}[1]{\overline{#1}}
\newcommand{\Div}[1]{\mbox{div}\left( #1 \right)}

\newtheorem{theorem}{Theorem}

\newtheorem{corollary}[theorem]{Corollary}
\newtheorem{lemma}{Lemma}
\theoremstyle{remark}
\newtheorem{remark}{Remark}
\theoremstyle{definition}
\newtheorem{definition}{Definition}

\title{Stability of Ordinary Differential Equations with Colored Noise Forcing}

\author{Timothy \ Blass$^{1}$\and  L.A.\  Romero $^{2}$}

\begin{document}
\footnotetext[1]{ Department of Mathematical Sciences, Carnegie Mellon University,
Pittsburgh, PA 15213, USA
{\tt tblass@andrew.cmu.edu}}
\footnotetext[2]
{Computational Mathematics and Algorithms Department,
Sandia National Laboratories, MS 1320, P.O. Box 5800, Albuquerque, NM 87123-1320,
{\tt lromero@sandia.gov}}
\footnotetext[3]
{Sandia National Laboratories is a multi-program laboratory
managed and operated by Sandia Corporation, a wholly owned subsidiary of
Lockheed Martin Corporation, for the U.S. Department of Energy's National
Nuclear Security Administration under contract DE-AC04-94AL85000.}

\maketitle

\begin{abstract}
We present a  perturbation method for determining the moment
stability of  linear ordinary differential equations with
parametric forcing by colored noise. In particular, the
forcing arises from passing  white noise  through an $n$th order filter.  
We  carry out a perturbation analysis based on a small parameter
$\varepsilon$ that gives the amplitude of the forcing.  
Our perturbation analysis is based on a
ladder  operator approach to the vector  Ornstein-Uhlenbeck
process. 
We  can carry out our perturbation expansion to any order in $\varepsilon$, 
for a large class linear  filters, and for quite arbitrary linear  systems.
As  an example we  apply our results to the stochastically forced  Mathieu equation.
\end{abstract}

{\bf Subject Class:} Primary: 93E15; Secondary: 60H10, 34D10.

\section{Introduction}
\subsection{ A Class of Stochastically Forced Linear Equations }

The original goal of this work was to develop a framework 
for analyzing  the stability of the stochastically forced 
Mathieu equation:
\begin{equation} \label{eq:Mathieu}
\ddot{x} + \gamma \dot{x} + (\omega_0^2 + \ep f(t))x = 0,
\end{equation}
where $f$ is a stochastic process, and the
stability is determined by the boundedness of the
second moment $\avg{x^2(t)}$ \cite{Arnold,khas}. 
Here, $\avg{\cdot}$ denotes
the sample-average. 
We wanted to avoid heuristic methods, and consider cases 
where $f(t)$ is a stochastic process with a 
realistic power spectral density. In particular, we do
not want to  assume that $f$ is white noise.  Hence
we want to analyze the case where $f(t)$ is colored noise. However,
in order to rigorously derive a Fokker-Planck equation for
a stochastic differential equation,   
the  governing equation must include only white noise 
\cite{Arnold}. 
We can achieve both goals of rigor and
realistic power spectral density by letting $f$ be the output of a linear
 filter that is forced by a vector white noise ${\boldsymbol \xi}$.
That is, 
\begin{align}
&\dot{\bfs} = \bfch \bfs +{\boldsymbol \xi}(t), \label{eq:filter}\\
& f(t) = \ats ,\label{eq:filter_f}
\end{align}
where $\bfch$ is an $n\! \times \! n$ real, diagonalizable matrix, whose eigenvalues
have negative real parts,  
$\bfa\in \R^n$, and
$\langle \cdot,\cdot\rangle$ is the
standard inner product on $\mathbb{C}^n$. 
We will take deterministic initial condition $\bfs (0)  = {\bf 0}$. 
We assume the noise vector ${\boldsymbol \xi}$ is  weighted white noise, meaning
\begin{equation}
  \label{eq:xi}
\avg{ {\boldsymbol \xi}}=0, \qquad
  \avg{{{\boldsymbol \xi}}(t+\tau)
{\boldsymbol \xi}^T (t)} = \bB \delta(\tau),
\end{equation}
where $\bfcb$ is symmetric and positive semi-definite.
Thus, when $\bfs (t)$ solves \eqref{eq:filter}
it is a standard vector-valued Ornstein-Uhlenbeck process. 
We refer to the scalar process,
$f(t) = \ats$, as colored noise or as an $n$th-order filter
provided $\bfs (t)$ solves \eqref{eq:filter}.
We will make only mild requirements on the matrices $\bfch$ and $\bfcb$,
thereby  allowing 
for wide variability in the power spectral density of the resulting
process $\ats$. Thus, 
in allowing for a wide range of 
choices of $\bfch, \bfcb$, and $\bfa$, 
our approach accommodates a broad class of colored noise
forcing terms.

In this paper we will be concerned with the more general problem 
of linear equations that are being parametrically forced 
by the function $f(t)$
in equation (\ref{eq:filter_f}).  That is equations of the form 
\begin{equation}\label{eq:x}
 \dot{\bx} = \azero \bx + \ep  \as  \aone \bx,
\end{equation}
where 
${\bf s}$ is  the solution to the stochastic equation
(\ref{eq:filter}), ${\bf x}(t) \in \mathbb{R}^N$ for
some $N\geq 1$,
and $\azero, \, \aone$ are $N\times N$ constant
matrices.

The purpose of this paper is to present a perturbation 
 method (assuming $\ep$  is small)  for determining the stability of
the solution $\bx(t)$ of \eq{x}, by which we mean the boundedness 
of the second moments of $\bx(t)$. However, our method applies to the $p$th
moment, so we will not limit our analysis to second moments only.
Van Kampen  has presented a heuristic approach to the 
case of colored noise forcing, \cite{Kampen}. Though derived by
completely different means, his result 
for the Mathieu equation \eq{Mathieu} is the same as
ours when considering only the first-moments, and without damping. 
He arrives at his  result by  truncating some 
series at order $\ep^2$, and 
is only expected to be valid to this order. Our method is rigorous,
can be applied to find solutions to any order in $\ep$, and applies
to any moment. We discuss this further in \S \ref{sec:compare}.

We were originally interested in 
equation (\ref{eq:Mathieu}) as a  model for the  response
of  capillary gravity  waves  to  a time-varying gravitational field
arising from  random vertical motions of a container  with a
 free surface (as in \cite{repetto}).
Here 
$f(t)$ represents  the random fluctuations in acceleration. 
Since the Fourier transform of an acceleration should
vanish at zero, along with its derivative, the power spectral 
density of 
a realistic 
process $f$  should satisfy $S(0) = S'(0)=0$.   
For example,   we can construct a  
two-dimensional filter using  the system (\ref{eq:filter}) 
that 
has the power spectral density
\begin{equation}
  \label{eq:exampleS}
  S(\omega) = \frac{ \sigma \beta^2 \omega^2}{(\omega^2+\mu_1^2)
(\omega^2+\mu_2^2)},
\end{equation}
by choosing
\[
{\bf H} = \left( \begin{array}{cc} -\mu_1 & 0 \\
\beta & -\mu_2 \end{array}\right),
\quad {\bf B} = \left( \begin{array}{cc} \sigma & 0 \\
0 & 0 \end{array}\right),
\quad \ba  = \left( \begin{array}{c} \beta \\
-\mu_2 \end{array}\right),\quad \mu_1,\mu_2,\sigma >0.
\]
The formula for $S(\omega)$ in equation \eq{exampleS}
follows from Corollary \ref{cor:S} in Appendix C.

 The stochastically forced Mathieu equation
has been analyzed before, for instance in
\cite{adams-bloch2008,adams-bloch2009,adams-bloch2010,bobryk-chrz,lee-etal,roy}
but not for the case \eqref{eq:Mathieu}, or the
general setting \eqref{eq:x}.
In  \cite{lee-etal, roy} they consider additive forcing, 
and in 
\cite{adams-bloch2008,adams-bloch2009,adams-bloch2010}
they consider a different type of parametric forcing.
In \cite{bobryk-chrz} they consider a different class of colored noises
and study stability by truncating an infinite hierarchy of moment equations.
Other studies concern Lyapunov stability, or rely on numerical methods.
Our analysis applies to a broad class of equations \eqref{eq:x} 
with a wide variety of forcing terms \eqref{eq:filter}, 
is semi-analytical (relying only on  numerics for the
 computation of eigenvalues
of small matrices),
 and can be applied to any moment.

\subsection{Ladder Operators and the Vector Ornstein-Uhlenbeck 
Process}

Our perturbation analysis of 
the moment stability of equation (\ref{eq:x}) relies heavily 
on a simple characterization of   the eigenvalues and eigenfunctions of the Fokker-Planck equation
associated with equation (\ref{eq:filter}).  
In particular, in \S \ref{sec:lad} and \S \ref{sec:eig}
 we characterize the spectrum using ladder operators
by generalizing
   Dirac's
creation and annihilation operator approach to the quantum harmonic oscillator \cite{Dirac}.
An understanding of the spectrum and eigenfunctions  in terms of its
ladder operators is crucial to developing the perturbation theory in \S \ref{sec:pert}.
Though 
other authors have used ladder operators for Ornstein-Uhlenbeck processes,
they have only considered the scalar case $n=1$ 
\cite{resibois,risken,titulaer,wilkinson}. 
We believe the extension to the vector case is by no means trivial, and
is interesting in its own right.

The probability density function $P(s_1, \ldots, s_n, t)$ 
associated with the process $\bfs(t)$  
defined by equations (\ref{eq:filter})   and (\ref{eq:xi}) 
satisfies the Fokker-Planck equation
\begin{equation}\label{eq:fps}
  \partial_t P = \D P ,
\quad \mbox{with}\qquad
  \D P = \frac{1}{2}\Div{\bfcb \nabla P} - \Div{\bfch \bfs P}.
\end{equation}
$\D$ is called the Fokker-Planck operator associated to \eqref{eq:filter}.
See \cite{Gardiner}  for a derivation of this equation. 
We note that the Fokker-Planck equation
\eq{fps} is the same in both the It\^o and Stratonovich interpretations because
the matrix $\bB$ is independent of $\bs$ (see \cite{Gardiner}). The operator
$\D$ will play a crucial role in our stability analysis.

In \S \ref{sec:lad} 
we begin by analyzing 
the operator $\D$ in terms of its associated ladder operators. That is,
operators $\mcl$ satisfying the commutator equation 
\begin{equation}\label{eq:ladEq}
 [\D, \mcl]=\mu \mcl 
\end{equation}
As in Dirac's theory of the harmonic oscillator, the significance
of the ladder operators stems from the fact that  if  $\phi$ is 
an eigenfunction of $\D$ with eigenvalue $\lambda$, then  the function
$\mcl \phi $ will either vanish, or be an eigenfunction of $\D$ with
eigenvalue $\lambda + \mu$.  

In \S \ref{sec:lad} we show that we can  construct the
ladder operators by  solving a matrix eigenvalue problem
\begin{equation}\label{eq:TDA}
  \bfct \bfy = \mu \bfy , \qquad \bfct = \bfcd \bfca,
\end{equation}
where $\bfca$ is an antisymmetric matrix and
$\bfcd$ is a symmetric matrix, expressed in terms of $\bfch$, $\bfcb$.
 We show there are $n$ raising operators, $\mcl_k$,
which generate new eigenfunctions of $\D$  with an increase in the real
part of the eigenvalue, and $n$ lowering operators $\mcl_{-k}$ that
correspondingly decrease the real part of the eigenvalue. 
We also  show that $\D$ can be expressed in terms of its
ladder operators. In particular,
\begin{equation}\label{eq:Dladder}
  \D = \sum_{k=1}^n \mu_k \mcl_{-k}\mcl_k.
\end{equation}
where $\mu_k$ is  the increment of the ladder operator $\mcl_k$. That
is, $[\D, \mcl_k] = \mu_k \mcl_k $. 
This representation is useful for determining the spectrum of $\D$.

In \S \ref{sec:eig} we characterize  the solutions of
\begin{equation}
  \label{eq:eig_prob}
  \D \phi = \chi \phi
\end{equation}
in terms of the ladder operators, $\mcl$, and increments, $\mu$, solving \eq{ladEq}.
In particular, we show that any eigenvalue $\chi$ of $\D$ can
be written as

\begin{equation}
\chi_{{\bf k} } = -\sum_{j=1}^n k_j \mu_j
\label{eqn_integ}
\end{equation}
where $\mu_j$ are the increments of the ladder operators with 
positive real parts, and the $k_j$ are non-negative integers.
We will see that the increments $\mu_j$ are the negative of the
eigenvalues of the matrix ${\bf H}$ defining the filter in
equation (\ref{eq:filter_f}).  
We also show that any eigenfunction of $\D$ can be obtained by
applying the ladder operator to the eigenfunction
$\Phi_0({\bf s})$
associated with the eigenvalue $\chi=0$ of $\D$,
which is the eigenvalue with the largest real part.  

The results summarized in the last paragraph
rely  on the fact that
real parts of the eigenvalues of $\D$ are bounded above (see
Lemma \ref{lem:boundD}) , which is proved 
in \cite{Liberzon, Metafune, roy},
but we give a different and simple proof of this in Appendix B.
Here, the domain of $\D$ is the  set of  functions 
that have bounded moments of any order.
The spectrum and eigenfunctions of $\D$ have been studied 
before (see \cite{Liberzon, Metafune, roy}) but not in the context of ladder operators.

\subsection{Perturbation Expansion for Moment Stability Analysis}

In  \S \ref{sec:pert} 
we use the classical perturbation theory of eigenvalues to carry 
out an analysis of the stability of equation (\ref{eq:x}).  
Our analysis begins by considering the ODEs for $\bfs (t)$ and $\bfx (t)$
together as a single ODE system.
The probability density function
$P(s_1,\ldots,s_n,x_1,\ldots,x_N,t)$ for 
the combined system \eq{filter} and \eq{x}
solves the Fokker-Planck equation 
\begin{equation} \label{eq:fp}
  \partial_t P = \frac{1}{2}\mbox{div}_\bs\left( \bB\nabla_\bs P\right) - \mbox{div}_\bs\left(
\bH \bs P\right) - \mbox{div}_\bx \left( (\azero + \ep  \as \aone)\bx P \right)
\end{equation}
The notation $\mbox{div}_\bs$ and $\nabla_\bs$ refer to divergence and gradient 
with respect to only the $s_j$ variables, and similarly
$\mbox{div}_\bx$ is divergence in $\bfx$ variables.
Equation \eq{fp}
is the same in both the It\^o and Stratonovich interpretations because
the matrix $\bB$ is independent of $\bfs$ and $\bfx$ (see \cite{Gardiner}). 

We  can derive an equation for the $p$th marginal moments by multiplying 
\eq{fp} by monomials $\bx^\alpha$ and integrating with respect to $d\bx$, where
$\alpha$ is a multi-index of order $p$. The result is an equation 
for $\bm (\bs,t)$,
 a vector of the $p$th marginal moments, 
which is of the form
\begin{equation}\label{eq:gzero}
  \partial_t \bm = \D \bm + \gzero \bm + \ep \as \gone \bm.
\end{equation}
Note that $\D$ is a differential operator in the $\bs$ variables only,
\begin{equation}\label{eq:defD}
\D \varphi  = \frac{1}{2}\mbox{div}_\bs\left( \bB\nabla_\bs \varphi \right) - \mbox{div}_\bs\left(
\bH \bs \varphi \right).   
\end{equation}
In equation (\ref{eq:gzero}) 
each component of $\bfm(\bfs,t)$ is of the form
$\int_{\R^N}\bfx^\alpha P(\bfx,\bfs,t)d\bfx$
for some multi-index $\alpha$ with $|\alpha|=p$.
 $\D \bm$ indicates $\D$ applied to each component of $\bm$. 
For much of our analysis  we can  assume that the matrices $\gzero$ and $\gone$
are given to us, but we illustrate how to obtain
 these matrices from
the matrices ${\bf A} _0$ and ${\bf A} _1$  for the particular case of 
the Mathieu equation in \S \ref{sec:app}.  
The matrices $\gzero, \, \gone$ in \eq{gzero} are constant
and depend on 
which moments one is considering (see example in equation \eq{gamma}).
There are $J=\left(
  \begin{array}{c}
    N+p-1 \\ p
  \end{array}\right)$ distinct $p$th order monomials in $N$ variables,
therefore $\gzero$ and $\gone$ are $J\times J$ matrices.

As in a standard stability analysis, 
in order to determine the stability of \eq{gzero}, we look for solutions of the form
 $\widetilde{\bfm}(\bs,t) = e^{\lambda t}\bfm(\bfs)$. Our equation for $\bm(\bs)$ becomes
\begin{equation} \label{eq:general}
     \lambda \bm = \D \bm + \gzero \bm + \ep \as \gone \bm.
\end{equation}
That is, the equation for the $p$th marginal moments of $\bx(t)$
can be written as an eigenvalue problem, and stability is
decided by the sign of the
real part of the largest eigenvalue.

We do a perturbation analysis assuming  that the magnitude $\ep$
of the forcing is small.  
Our analysis relies on the fact that that when $\ep=0$ 
the eigenfunctions
of equation (\ref{eq:general}) are the direct product of the eigenfunctions
of $\D$ and the eigenvectors of the matrix $\gzero$.  


A key observation (see Lemma \ref{lem:alpha_beta})  for the perturbation analysis is that
for any vector ${\bf a}$ we can  determine constants $\alpha_k$ and $\beta_k$ 
such that 
$\as$   can be written as
\begin{equation}\label{eq:coeff_intro}
\as = \sum_{k=1}^n \left(\alpha_k \mcl_k + \beta_k \mcl_{-k}\right),
\end{equation}
where $\mcl_{\pm k}$ are the ladder operators satisfying \eqref{eq:ladEq}.
The proof of Lemma  \ref{lem:alpha_beta} is given in Appendix D.  

In \S \ref{sec:pert} we show that 
when $\ep=0$, the eigenvalue of equation (\ref{eq:general})  with the
largest real part is
the same as the largest eigenvalue of the matrix $ \gzero$.
If $\lambda_0$ is this unperturbed eigenvalue,
then $  \lambda(\ep)  = \lambda_0 + \lambda_2 
\ep^2 +\ldots$ with

\begin{equation}\label{eq:lamb2}
\lambda_2 = \sum_{j=1}^J \ip{\bfpsi_1}{\gone \bfphi_j}\ip{\bfpsi_j}{\gone \bfphi_1}
G(\nu_1 - \nu_j) 
\end{equation}
where $\nu_1 = \lambda_0$, $\nu_j$ are the eigenvalues of $\gzero$,
and $\bfphi_j , \bfpsi_j$ are the eigenvectors and
normalized adjoint eigenvectors of $\gzero$. 
Equation (\ref{eq:lamb2}) 
uses   the \emph{extended power
  spectral density} $G(z)$, which is defined for a general 
stationary random process  in \eq{Gdef},
and is given explicitly in \eq{G} for the filter 
$\ats$.

The form of $\lambda_2$ 
in \eq{lamb2} is derived for forcing terms that have the form \eq{filter_f}, 
however, the fact that this  is simply a weighted sum of 
values of $G$, whose coefficients  depend only on $\gzero, \gone$
(which do not depend on the filter),
suggests such a formula could hold for any process with a well-defined
extended power spectral density. We have carried the perturbation analysis to
higher orders, but the higher-order
coefficients do not appear to have such a simple
form as in equation \eqref{eq:lamb2}. 

The method in \S \ref{sec:pert} 
involves constructing matrices $\gzero$ and $\gone$,
which, as mentioned earlier,
depend on $\azero, \aone$, and the  representation of the $p$th
marginal moments as a vector.
We use the stochastic Mathieu equation as a specific example
 in \S \ref{sec:app}.
 In \S \ref{sec:numerics} we   discuss a numerical method
for determining the stability of \eq{x}  without assuming that
$\varepsilon$ is small.  We  compare these numerical results to our
perturbation results up to both second and fourth order 
for the Mathieu equation, and show they are in excellent agreement. In
\S \ref{sec:tensor} we give a second representation (whose
derivation is given in \cite{SAND}) 
for $\lambda_2$ that does not involve the  matrices $\gzero$ and
$\gone$, but deals directly with the matrices ${\bf A} _0$ and ${\bf A} _1$.
We have found that this representation
 simplifies numerical computations.

\section{Existence and Properties of Ladder Operators}\label{sec:lad}
In this section we define  the 
 notion of a ladder operator,
show how to construct these operators,   and 
prove some basic lemmas about them.
Lemma \ref{lem_one} shows how ladder operators can be used to generate new 
eigenfunctions that have their eigenvalue changed by 
the increment of the ladder operator.
Lemma \ref{matrixEig} shows how to 
find the ladder operators $\mcl_k$ and their  increments $\mu_k$   by
solving a matrix eigenvalue problem. 
Lemma \ref{lem:T} shows that the increments of  the ladder
 operators are
zero and 
$\pm \mu_k$ ,where  $- \mu_k$ are the eigenvalues of the matrix 
${\bf H}$ defining our filter (see equation (\ref{eq:filter})).
 Lemma \ref{lem_commute} gives the commutator relations for the ladder 
operators, and Lemma \ref{lem:Dsum} shows that the operator $\D$ can
be expressed as  a  weighted sum of  $\mcl_{-k} \mcl_k$, where 
$\mcl_k$ are the ladder operators.
Throughout this section  (and the rest of the paper)
  the operator $\D$  is that defined in \eqref{eq:defD}.

The basic lemmas in this section  are 
 crucial to the rest of this paper, and hence  we have
tried to write this section so that 
the lemmas stand out clearly.  Though the lemmas are all easily
stated,  the proofs of some of the lemmas are 
quite technical, especially when attention is given to ensuring that they
apply for complex eigenvalues of the matrix $\bfct$.
For this reason, we have relegated many of the proofs to Appendix A.

Before discussing ladder operators it should be 
noted that we define
the domain of $\D$ as the  set of  functions 
that have bounded moments of any order.
Thus, our definition of the domain of $\D$
  differs from that given
in \cite{Liberzon}.  In that paper they defined the domain based on the
exponential decay
of the eigenfunction $\Phi_0$ that we define
in equation (\ref{eqn_defPhi0}) and discuss in \S \ref{sec:eig}.
The two definitions of the domain give the identical eigenfunctions,
but we believe ours is more natural since it does not require knowing
the solution ahead of time.
In \cite{Liberzon} they discuss a continuous
spectrum that arises if the domain is defined  so that the eigenfunctions
of $\D$ are only required to be square integrable (or some similarly
less restrictive condition).  
An examination of these eigenfunctions shows that 
they have a power law decay as ${\bf s}$ goes to 
infinity, and hence do not have moments of all orders.  
Hence our definition of the domain also excludes this continuous spectrum.  

We now give a definition of a ladder operator of $\D$.  
\begin{definition}
  An operator $\mcl$ is a ladder operator for $\D$ with increment $\mu$ 
if $[\D , \mcl ]=\mu \mcl$
for some $\mu \in \C$, where $[\cdot \, , \cdot ]$ denotes the commutator
$[\D,\mcl] = \D \mcl - \mcl \D$.
\end{definition}
The following lemma shows that ladder operators can be used to
generate new eigenfunctions from   ones that we already know.
\begin{lemma}
\label{lem_one}
Suppose $\mcl$ is a ladder operator such that 
$[\D,\mcl ] = \mu \mcl$.  Let $\phi$ be an eigenfunction of
$\D$ with eigenvalue $\chi$.  Then either $\mcl \phi=0$, or
$\mcl \phi$ is an eigenfunction of $\D$ with eigenvalue 
$\chi + \mu$.  
\end{lemma}
\begin{proof}
We have $\D \mcl \phi - \mcl \D \phi = \mu \mcl \phi$.  
Since $\phi$ is an eigenfunction of $\D$, this gives us
$\D \mcl \phi = (\chi + \mu ) \mcl \phi $.  
\end{proof}

We defined the domain of $\D$ to be the set of functions
that have moments of all orders.  
It should be noted that Lemma \ref{lem_one} would not apply 
if the domain  had been (for example) the set of all  square integrable 
functions.  
In that case a third possibility would exist.  It could be
 that   the function $\phi$ is square integrable, but the
function  $\mcl \phi$  is not.  Thus $\mcl \phi$ would not generate
a new eigenfunction.

We will show that $\D$ has $2n+1$ ladder operators $\mcl_{\pm k}$, $k=0,\ldots,n$.
 We begin by decomposing $\D$ into simple differential
and multiplicative operators. 

\begin{definition}\label{def:L}
We define the operators $L_j, j=1,\ldots, 2 n+1$ as follows.
\begin{equation}
\label{eqn_defL}
\begin{aligned}
& L_j \phi  = \partial_{s_j} \phi \;\;\;\;\; \mbox{for $j=1,\ldots , n$} \\
& L_{j+n} \phi =              s_j \phi \;\;\; \mbox{for $j=1,\ldots , n$}\\
& L_{2n+1} \phi =             \Id     \phi\;\;\;\; \mbox{for $j=2n+1$}
\end{aligned}
\end{equation}
\end{definition}
Here $\Id$ is the identity operator.  
Note that 
$[L_j,L_k]=0$ unless $|j-k|=n$, and $[L_j,L_{j+n}]= \Id $.
In particular, we have
\begin{equation}
[ L_j, L_{k+n} ] = \delta_{jk} \Id \;\;\;\mbox{   $j,k=1,\ldots,n$}
\label{eqn_comL}
\end{equation}
We note that $\D$ can be expressed in the operators
$L_j$ as 
\begin{equation} \label{eq:Dsum}
  \D = \sum_{k,j=1}^{2n+1}\frac{1}{2}d_{jk}L_jL_k.
\end{equation}
We let $\bfcd$ denote the symmetric matrix with components
$d_{jk}$ in \eqref{eq:Dsum}. The choice of $d_{jk}$ 
in \eqref{eq:Dsum} is not unique, but we make an
explicit choice that makes this matrix symmetric.
If we let $\bfca$ denote the antisymmetric matrix 
with components  $a_{jk}$ given by 
\begin{equation} 
[L_j,L_k] = a_{jk} \Id  ,
\label{eqn_defA}
\end{equation}
then we have  explicit expressions  for $\bA$ and $\bD$ 
\begin{equation}
  \label{eq:AD}
\bA = \left(\begin{array}{ccc} 
\bZero_n & \bI_n & 0 \\ -\bI_n & \bZero_n & \vdots \\ 0& \cdots &0\end{array} \right),
\quad  
\bD = \left( \begin{array}{ccc} \bB & - \bH & 0\\
-\bH^T & \; \bZero_n & \vdots \\
0 & \ldots &  - \mbox{tr}(\bH) \end{array} \right).
\end{equation}
For details regarding the construction of $\bfcd$ see
Lemma \ref{Dsymm} in Appendix A. 

Just as $\D$ has a representation in terms of the
operators $L_j$, its 
ladder operators will also be expressed in terms of the $L_j$. 
Consider an operator
\begin{equation}\label{eq:L}
\mcl = \sum_{j=1}^{2n+1}y_jL_j.
\end{equation}
We write $\bfy$ for the vector of coefficients of $\mcl$. 
From the representations \eqref{eq:Dsum} and \eqref{eq:L},
we see that   the commutator $[\D,\mcl]$  involves
sums of terms of the form $L_iL_jL_k - L_k L_i L_j$,  which do not
at first sight appear to be linear in the operators  $L_m$, 
$m=1,\ldots,2n+1$.
 However, by twice applying the
commutator relations in equation (\ref{eqn_comL}) we can show that
$L_i L_j L_k - L_k L_i L_j$ is in fact 
 a sum of the $L_m$.
Determining the coefficient vector $\bfy$
and increment $\mu$  thus becomes a matrix eigenvalue problem.
The details of how we arrive at this form are 
given in Appendix A.  Here we will merely state the result of
these manipulations.  
\begin{lemma} \label{matrixEig}
If $\bfy$ is the vector of coefficients for $\mcl$, as defined as in equation \eq{L}, 
then the equation $[\D,\mcl]=\mu \mcl$ can be written as 
a matrix eigenvalue problem
$\bfct \bfy = \mu \bfy$, where $\bfct = \bfcd \bfca$, 
and $\bfcd$  and $\bfca$ are
defined in \eqref{eq:AD}.
\end{lemma}

We make the assumption that the eigenvalues of $\bfch$ have negative real parts,
and the eigenvectors form a complete set. For simplicity of the arguments, we will
also assume that the eigenvalues of $\bfch$ are simple.
By explicitly writing out the eigenvalue problem
$\bfct {\bf y} = \mu {\bf y}$ we can determine the eigenvalues $\mu_k$
in terms of the eigenvalues of the matrix ${\bf H}$.  
We will give the details of the proof in Appendix A.  

\begin{lemma} \label{lem:T}
  The eigenvalues of $\bfct =\bfcd \bfca$ are $\{0,\pm \mu_k \}$, $k=1,\ldots , n$, 
where $-\mu_k$ are the eigenvalues of the matrix ${\bf H}$.
\end{lemma}
Note that $\mcl_0$ is the identity operator with increment $0$.
Thus, our analysis only involves the $2n$ ladder operators $\mcl_{\pm k}$
for $k=1,\ldots, n$.

In doing the perturbation expansion it will be necessary
to have the commutator relations of the operators $\mcl_i$.
Finding the commutator relations for $[\mcl_j,\mcl_k]$ can be 
turned into a linear algebra  problem involving the eigenvectors of
the matrix $\bfct$.
In particular,  using equations  (\ref{eq:L}) and (\ref{eqn_defA}) we 
get

\begin{equation}
[\mcl_j,\mcl_k ] =  \left({\bf y} _j ^T {\bf A} {\bf y} _k \right) \Id.
\label{eqn_com1}
\end{equation}
From equation (\ref{eq:AD}) it is easily seen that
\begin{equation} 
{\bf A} {\bf D} =
-({\bf D} {\bf A} )^T  = - {\bf T} ^T 
\label{eqn_antic}
\end{equation} 
If ${\bf D} {\bf A} {\bf y} = \mu 
{\bf y}$, then multiplying  both sides of this by ${\bf A}$ and using
equation (\ref{eqn_antic}) we see that ${\bf A} {\bf y} $ is an
eigenvector of ${\bf T} ^T$ with eigenvalue $ - \mu$.  
With this in mind, the right hand side of equation (\ref{eqn_com1}) can 
be written as the inner product between the vector ${\bf y} _j$  and the
adjoint eigenvector of ${\bf T}$  associated with $ - \mu_k$.  
Using the fact that the eigenvectors and adjoint eigenvectors of a 
matrix form a bi-orthogonal set, we can arrive at a simple
expression for the commutators.   

When dealing with complex quantities, the notation in this
argument gets to be a bit tedious,  and we will leave the details 
to Appendix A.  
The final commutator result is  given by the
following lemma.

\begin{lemma} \label{lem_commute}
  For $j,k \geq 1$ we have $[\mcl_j,\mcl_k] = 0$ and $[\mcl_{-j},\mcl_{k}] = 
 \delta_{jk} \Id$.
\end{lemma}


In Dirac's theory of the  harmonic oscillator, he shows that
the  Hamiltonian operator can be written as the 
product of the raising and  lowering operators. 
We now generalize this result 
to  the vector case.  In this case the operator $\D$ can be
written as a weighted sum of the products of the raising and
lowering operators.  
The next lemma shows that the weights are in fact the 
eigenvalues $\mu_k$ of the matrix $\bfct$.
We leave the proof of this  lemma to Appendix A, but note
that its proof is probably the most subtle one in this
paper.

\begin{lemma} \label{lem:Dsum}
The differential operator  $\D = \sum_{i,j=1}^{2n+1}\frac{1}{2}d_{i,j}L_iL_j$ 
can be written as 
\begin{equation}
  \label{eq:Dlad}
\D = \sum_{k=1}^n \mu_k \mcl_{-k}\mcl_{k}.
\end{equation}
\end{lemma}
An important feature of the decomposition
\eqref{eq:Dlad} is that only terms of the form $\mcl_{-k}\mcl_{k}$, $k>0$, appear (there
are no terms of the form $\mcl_{k}\mcl_{-k}$ for $k>0$).

\section{Eigenvalues and Eigenfunctions of $\D$}\label{sec:eig}
In this section we  will use the ladder operator formalism to completely characterize the 
eigenvalues and  eigenfunctions of the  operator $\D$.  
We note that the spectrum of $\D$ has already been studied and
characterized \cite{Liberzon,Metafune,roy}, but not in terms of ladder operators. 
We include another proof of those results because the characterization in
terms of ladder operators is used in the perturbation analysis in \S \ref{sec:pert}.


As  with Dirac's theory of the quantum harmonic oscillator, the
analysis of the spectrum using ladder  operators requires that the
real part of the spectrum is bounded above.  
We will now state  this as a lemma, but leave the proof to
 Appendix B.  

\begin{lemma}\label{lem:boundD}
  The real part of spectrum of the operator $\D$, as defined in \eqref{eq:defD}, is bounded above.
\end{lemma}

The following theorem will  allow us to characterize the eigenfunction
associated  with the eigenvalue with the largest real part.

\begin{theorem}
\label{thm:lkphi}
Let $\Phi({\bf s})$ be an eigenfunction  of $\D$ (as in
equation \eqref{eq:defD}) associated with the 
eigenvalue having the largest real part.
We must have $\mcl _k \Phi =0 $ for $k=1, \ldots, n$.  
\end{theorem}
\begin{proof}
Suppose $\Phi({\bf s})$ is an eigenfunction of $\D$ with
eigenvalue $\chi$.  If $ \Psi = \mcl_k \Phi \neq 0$, then  $\Psi ({\bf s})$
will be an eigenfunction of $\D$ with eigenvalue $\chi + \mu_k$.  
This will give us an eigenvalue with a larger real part than $\chi$. 
 Hence if $\chi$ is the eigenvalue with the largest real part, then
$\mcl_k \Phi =0$ for all $k$.  
\end{proof}

\begin{remark}
The system  $\mcl_k \Phi = 0$ for each $k=1,\ldots, n$ is an over-determined system of
first order differential equations. The fact that a solution exists is non-trivial.
However, the fact that $[\mcl_k,\mcl_j] =0$ implies
that the Frobenius Theorem applies (see \cite{Marsden, ArnoldV}), which
guarantees the system is solvable.
\end{remark}
\begin{remark}
As in the comment following Lemma \ref{lem_one}, we should note that
the domain of $\D$ is  defined as the set of functions that 
have moments  of all orders.  
If  the domain of $\D$ were defined using the less stringent
requirement that the eigenfunctions were square integrable, it would not
be necessary that $\mcl _k \Phi=0$ for all $k$.  This is because 
in this case
$\mcl_k \Phi$ does not have to generate 
a new eigenfunction.  It could instead  produce a function that 
is not square integrable.  
\end{remark}

By Theorem \ref{thm:lkphi}, the ``top'' eigenfunction $\Phi_0({\bf s})$ (i.e.
the eigenfunction associated to the largest eigenvalue of 
$\D$) must satisfy the
equations $\mcl _k \Phi _0 =0$.
If ${\bf y} ^k$ is the eigenvector of $\bfct$ associated with
the eigenvalue $\mu_k$, and if $\mu_k \neq 0$, then
the last  component of ${\bf y} ^k$ vanishes (see the proof of
Lemma \ref{lem:T} in Appendix A). 
That is, we can write 
\begin{equation}
{\bf y}^k = \threevec{{\bf p} _k}{{\bf q} _k}{0}.
\end{equation}
Using  equation (\ref{eq:L}),
and  the definition of the  operators $L_k$ in (\ref{eqn_defL}),
the equations $\mcl_k \Phi_0=0$  can thus  be written as

\begin{equation}
\label{eqn_lcmPhi0}
{\bf p} _k \cdot  \nabla \Phi_0 + ( {\bf q} _k \cdot {\bf s} ) \Phi_0 =0
\;\;\;\;k=1,\ldots,n
\end{equation}

If we make the  ansatz  that $\Phi_0({\bf s}) = 
\exp( - \frac{1}{2} {\bf s} ^T \boldsymbol \Sigma  {\bf s} )$, then
equations (\ref{eqn_lcmPhi0}) will be satisfied if and only if 
\begin{equation}
{\bf P} ^T \boldsymbol \Sigma = {\bf Q} ^T ,
\end{equation}
where 
\begin{equation}
{\bf P} = [{\bf p} _1,{\bf p} _2,\ldots,{\bf p} _n],\qquad
{\bf Q} = [{\bf q} _1,{\bf q} _2,\ldots,{\bf q} _n].
\end{equation}
If ${\bf P}$ is invertible, this gives us 
$\boldsymbol \Sigma  = ({\bf P} ^T)^{-1} {\bf Q} ^T $.  
It  is not clear that  ${\bf P}$ is invertible, or that
$\boldsymbol \Sigma $ is symmetric.  However, under
certain weak assumptions 
on ${\bf H}$ and ${\bf B}$ (see Definition \ref{def:control} 
and Lemma \ref{lem:Sigma} below)
this  will be the case.  If these  assumptions hold, it is convenient to
write $\boldsymbol \Sigma^{-1} = {\bf P} {\bf Q} ^{-1}$.   
We now define the notion of a controllable pair.  
\begin{definition}\label{def:control}
The matrices  ${\bf H}$  and ${\bf B}$  will be said to form
a controllable pair if 
there is no nontrivial vector ${\bf z}$ such that
$ {\bf z} ^T {\bf H} ^{k} {\bf B} =0$ for $k=0,\ldots,n-1$.
This is equivalent to requiring ${\rm rank}\,{\bf C} = n$, where
${\bf C}$ is the $n\times n^2$ matrix
${\bf C}=\left[{\bf B},\,{\bf H}{\bf B},\, \ldots, \, {\bf H}^{n-1}{\bf B}  \right]$.
\end{definition}
In Appendix B  we prove the following lemma.  

\begin{lemma}\label{lem:Sigma}
Assuming all of the eigenvalues of ${\bf H}$ have real parts less than
zero, the eigenvectors of $\bfch$ are complete, 
and that ${\bf B}$ is positive semidefinite, 
then    ${\boldsymbol \Sigma} ^{-1} = {\bf P} {\bf Q} ^{-1} $ 
is symmetric and positive semi-definite.  
If ${\bf H}$ and ${\bf B}$ also form a controllable pair, then
the matrix ${\bf \Sigma}^{-1}$ is positive definite, and hence
the matrices  ${\bf \Sigma} = {\bf Q} {\bf P} ^{-1}$ and ${\bf P}$
are non-singular.
\end{lemma}
 Requiring $({\bf H},{\bf B})$ to be a controllable pair
eliminates some ``degenerate'' types of filters.
For instance, if ${\bf H}={\rm diag}(-\mu_1,-\mu_2)$
and ${\bf B}={\rm diag}(1,0)$ then $({\bf H},{\bf B})$ is not a controllable
pair. In this case,
$s_1(t)$ is a scalar Ornstein-Uhlenbeck process, but $s_2(t)$ is deterministic,
so $\bfs (t)$ is not a genuine two-dimensional Ornstein-Uhlenbeck process,
but rather it is a one-dimensional process with an appended deterministic component. 

\begin{definition}\label{def:bc}
  We will say that the $n\times n$ real matrices
${\bf H}$ and ${\bf B}$
satisfy the \emph{basic conditions} if
\begin{itemize}
\item[(i)] ${\bf B}$ is symmetric
and positive semi-definite
\item[(ii)] ${\bf H}$ has simple eigenvalues
$\{-\mu_k\}_{k=1}^n$ with $\rp{\mu_k}>0$ for $k=1,\ldots,n$
\item[(iii)] $({\bf H},{\bf B})$ form a controllable pair (Def. \ref{def:control}).
\end{itemize}
The requirement of simple eigenvalues for ${\bf H}$ is for convenience
and could be replaced with the requirement of a complete
set of eigenvectors.
\end{definition}

\begin{lemma}
\label{lem:phi0}
Assuming ${\bf H}$ and ${\bf B}$  satisfy the basic conditions (Def. \ref{def:bc}),
the eigenvalue $\chi_0 $  with the largest real part  of $\D$ is simple, and
the  eigenfunction $\Phi_0({\bf s})$ associated  with
it is given by 
\begin{equation}
\label{eqn_defPhi0}
\Phi_0({\bf s}) 
= {\rm exp}\left( - \frac{1}{2} \langle{\bf s} , {\bf \Sigma} {\bf s}\rangle\right),
\end{equation}
where $\boldsymbol \Sigma  = {\bf Q} {\bf P} ^{-1} $. 
Moreover, $\chi_0 = 0$.
\end{lemma}
\begin{proof}
Without loss of generality we look for solutions of the form
$\Phi_0({\bf s}) = e^{ \psi_0({\bf s}) }$.  
In order  to satisfy equations (\ref{eqn_lcmPhi0}) we must have
\begin{equation}
{\bf p} _k \cdot \nabla \psi_0 + ({\bf q} _k \cdot {\bf s}) =0
\end{equation}
A  direct calculations shows that $\psi_0 = - \frac{1}{2} \langle
{\bf s}, \boldsymbol \Sigma {\bf s} \rangle$ satisfies this equation. 
If  we have another solution to this equation, say $\psi_1$, then the difference
$\psi_0- \psi_1$ between these solutions will
satisfy
$ {\bf p} _k \cdot  \nabla (\psi_0- \psi_1) =0$, for $k=1,\ldots,n$.  
The vectors ${\bf p} _k$ are complete
(they are the eigenvectors of ${\bf H}^T$),
 which implies that
$\psi_0-\psi_1$ is a constant.  This in turn implies that
the eigenfunctions associated  with each of the
solutions 
$\psi_0({\bf s})$  are   multiples of  each other, hence  
$\chi_0$ is simple.

From Lemma \ref{lem:boundD}, the real part of the the spectrum of $\D$ is 
bounded above.
From Lemma \ref{lem:Dsum}, we can write $\D = \sum_{k=1}^n \mu_k \mcl_{-k}\mcl_k$.
Hence,  when we apply $\D$ to the eigenfunction $\Phi_0({\bf s})$
associated with  $\chi_0$, 
 we will get  $\D \Phi_0 = 0$ 
because $\mcl_k \Phi_0=0$, $k=1,\ldots,n,$
by Theorem \ref{thm:lkphi}. Hence $\chi_0 = 0$.
\end{proof}

\begin{theorem}\label{thm:integ}
Let ${\bf H}$ and ${\bf B}$  satisfy the basic conditions (Def. \ref{def:bc}).
$\chi$  is an  eigenvalue of $\D$ if and only if it can be written as
in equation (\ref{eqn_integ}),
where $k_j$ are non-negative integers, and $- \mu_j$ are the eigenvalues of
${\bf H}$.
The eigenfunction associated with 
$\chi_{{\bf k}}$ is given by 
$\Phi_{{\bf k}}({\bf s}) = \mcl_{-1}^{k_1} \mcl_{-2}^{k_2}
...\mcl_{-n}^{k_n} \Phi_0({\bf s}) $.
where $\Phi_0({\bf s})$ is defined  as in equation (\ref{eqn_defPhi0}).
\end{theorem}
\begin{proof}
From Lemmas \ref{lem:boundD} and  \ref{lem:phi0}, the real part of the the spectrum of $\D$ is 
bounded above by $0$, and $\chi = 0$ is an eigenvalue of $\D$ that has the
form \eqref{eqn_integ}.
If $\chi$ is any other eigenvalue, 
and $\Phi$ is its eigenfunction, then  there must be at least one
value of $k$ such that $\mcl _k \Phi \neq 0$.  
If this new eigenvalue  has  the form 
given in equation (\ref{eqn_integ}), then the previous one will too. We
can keep carrying out this process obtaining eigenvalues with 
larger real parts.  This process must eventually end since the
real part of the spectrum is bounded
above.  The only way it can end is when we arrive at the largest eigenvalue,
which we have already seen, is zero.  This implies
equation (\ref{eqn_integ}). 

The argument in the last paragraph shows that any eigenvalue of
$\D$ must be of the form  (\ref{eqn_integ}).  
To show that any number $\chi_{{\bf k}}$ of this form 
must be an eigenvalue of $\D$ we show that 
 $\Phi_{{\bf k}}({\bf s}) = \mcl_{-1}^{k_1} \mcl_{-2}^{k_2}
...\mcl_{-n}^{k_n} \Phi_0({\bf s})$ is
the eigenfunction associated with 
$\chi_{{\bf k}}$.  This  
 follows from
Lemma \ref{lem_efuncs} in Appendix B.
\end{proof}

\section{Perturbation Method}\label{sec:pert}
The marginal-moment equation \eq{gzero} 
is derived by multiplying \eq{fp} by a monomial
$\bfx^\alpha$ for some multi-index $\alpha$, then integrating with
respect to ${\bf x}$.
If this is done for each multi-index of order $p$, we derive a set of equations
for the $p$th marginal moments. If we collect the $p$th marginal moments into a
vector $\bfm$, we arrive at \eq{gzero}. The matrices $\gzero ,\,\gone$ 
depend not only on $\azero, \aone$, but also on our mapping of the $p$th marginal
moments into $\bfm$. For this reason, we do not write the explicit form of
$\gzero, \gone$ in this section, but we do write them out for the example of
second marginal moments  for the Mathieu equation in \S \ref{sec:app}.

We let $\bfphi_j$ denote the eigenvectors of $\gzero$ with 
eigenvalues $\nu_j$. 
We let $\bfpsi_j$ be the normalized
adjoint eigenvectors, so that $\ob{\bfpsi}_k^T\bfphi_j = \ip{\bfpsi_k}{\bfphi_j} = \delta_{kj}$. 
We may assume without loss of generality that the $\nu_j$ are ordered
so that $\rp{\nu_1} \geq \rp{\nu_j}$ for all $j$.

We expand the unknowns as series in $\ep$, 
\begin{equation}\label{eq:expand}
\lambda = \lambda_0+\ep\lambda_1 +  
\ep^2 \lambda_2 + \ldots ,\qquad \bfm(\bfs) = 
\bfm_0(\bfs) + \ep \bfm_1(\bfs)+ \ep^2 \bfm_2({\bf s})  +\ldots,  
\end{equation}
and solve for the terms of these series. 
If we substitute these expansions into   \eq{general},
and collect 
the zeroth-order  terms, we get
\begin{equation}\label{eq:zero}
  \lambda_0 \bm_0 = \D \bm_0 + \gzero \bm_0.
\end{equation}
The eigenfunctions of $\D$ are scalar-valued, and the eigenvectors of $\gzero$ 
are constant vectors.  Assuming that both the eigenfunctions of
$\D$  and the 
eigenvectors of $\gzero$ are complete, then 
the most general 
solution $\bfm_0$ to \eq{zero} will be a product of an eigenfunction of $\D$
with an eigenvector of $\gzero$, and $\lambda_0$ will be the sum of the eigenvalues 
of $\D$ and $\gzero$.
We are interested in the largest eigenvalue, so we take
\begin{equation}
  \label{eq:m0}
\bm_0(\bfs) =  \Phi_0(\bfs) \bfphi_1,
\end{equation}
and $\lambda_0 = \nu_1$ because $0$ is the largest eigenvalue of $\D$
 and $\D \Phi_0 = 0$ (Lemma \ref{lem:phi0}), and $\nu_1$ was selected to have  the 
largest possible real part (note that the choice of $\nu_1$ need not
be unique).

The form of the forcing in \eqref{eq:x} allows us to
represent $\as$ in terms of the ladder operators.
In particular, the parametric forcing by the linear filter
$\ats$ results in the presence of the first-order
polynomial $\as$ in the Fokker-Planck equation, and thus to the 
term $\ep \as \gone {\bf m}$ in the moment equation \eqref{eq:general}.
Since the ladder operators, $\mcl_{\pm k}$, are 
linear combinations of first-order operators
$\partial_{s_j}$ and monomials $s_j$,
it is reasonable to try 
to write $\as$ as a linear combination of $\mcl_{\pm k}$.
The completeness of the
eigenvectors of $\bfch$ allows us to do this,
greatly simplifying our perturbation analysis.
\begin{lemma}\label{lem:alpha_beta}
If the eigenvectors of $\bfch$ are complete, 
and $\alpha_k, \, \beta_k$ are defined as in 
equation \eqref{eq:coeffs1} (Appendix D), then
\begin{equation}\label{eq:as}
\as = \sum_{k=1}^n \left(\alpha_k \mcl_k + \beta_k \mcl_{-k}\right).
\end{equation}
\end{lemma}
The proof of Lemma \ref{lem:alpha_beta} is in Appendix D.
Formula \eqref{eq:as} ensures that the
coefficients $\alpha_k$ and $\beta_k$  will appear in the
coefficients of the perturbation expansions \eqref{eq:expand}.
We show in Appendix C that the
extended power spectral density, $G(z)$, of $\ats$ can also be expressed
in terms of $\alpha_k$ and $\beta_k$
(Theorem \ref{thm:G}). This allows us
to derive a simple formula for the
order $\ep^2$ coefficient of $\lambda(\ep)$
in terms of $G(z)$ (Theorem \ref{thm:m2}).

Recall that there are
$J=\left(
  \begin{array}{c}
    N+p-1 \\ p
  \end{array}\right)$ distinct $p$th order monomials in $N$ variables,
and that $\gzero$ and $\gone$ are $J\times J$ matrices.
We will assume that $\gzero$ has a complete set of eigenvectors,
which is the case for the Mathieu equation, and occurs whenever
the eigenvectors of $\azero$ are complete.
The following lemma gives  solvability conditions that 
 will be used  repeatedly in our analysis.
\begin{lemma}\label{lem:solve}
Let ${\bf H}$ and ${\bf B}$ satisfy the basic conditions (Def. \ref{def:bc}).
Suppose that $\gzero$  has a complete set of
eigenvectors 
 $\{ \bfphi_j \}_{j=1}^J$,
 with eigenvalues $\nu_j$, 
normalized adjoint eigenvectors 
 $\{ \bfpsi_j \}_{j=1}^J$,
and that 
$\Phi(s)$ is an eigenfunction of $\D$ with eigenvalue
$- \mu $.
If $\mu \neq 0$, then
then  the equation 
\[
(\lambda_0 -\D - \gzero)\bm = \Phi(\bs)\bb
\]   
has a solution given by
\begin{equation}
  \label{eq:sol1}
\bm(\bs) = \Phi(\bs) \sum_{j=1}^J \frac{\ip{\bfpsi_j}{\bb}}{\nu_1-\nu_j +\mu} \bfphi_j 
\end{equation}
On the other hand, if $\mu=0$ (and hence  
 $\Phi(s) = \Phi_0({\bf s} )$) and  $\nu_1 \neq \nu_j$ for $j>1$, then the equation
\[
(\lambda_0 -\D - \gzero)\bm = \Phi_0(\bs)\bb
\]   
has a solution if and only if $\ip{\bfpsi_1}{\bb}=0$. In this case,
the solution is
\begin{equation}
  \label{eq:sol2}
\bm(\bs) =  \kappa \Phi_0(s) \bfphi_1 + 
\Phi_0(\bs) \sum_{j=2}^J \frac{\ip{\bfpsi_j}{\bb}}{\nu_1-\nu_j} \bfphi_j .
\end{equation}
where $\kappa$ is an arbitrary constant.  
\end{lemma}
The constant $\kappa$ can be used to choose a normalization for ${\bf
  m}$. We do not need to choose a specific normalization for ${\bf
  m}$, so we set $\kappa = 0$ because it is convenient. One can check
that if $\azero$ has a complete set of eigenvectors then $\gzero$
will too.
\begin{proof}
  If $\mu \neq 0$, then when we write ${\bf b}$ in the $\bfphi_j$ basis,
and make the ansatz $\bm(\bs) = \Phi(\bs)\bc$,
where $\bc$ is a constant vector, we arrive at the
expression for $\bm$ in equation (\ref{eq:sol1}).  If $\mu=0$, and hence $\Phi({\bf s}) 
= \Phi_0({\bf s})$, then we cannot solve this equation if
${\bf b}$ has any component in the direction of $\bfphi_1$.   
This gives the compatibility condition $\ip{\bfpsi_1 }{{\bf b}}=0$.  
Assuming this holds, the solution is given by equation (\ref{eq:sol2}).  
\end{proof}
We will now describe the outline of the perturbation analysis.
In order to help us describe the perturbation  analysis we will use the 
following definition.
\begin{definition}
We say a function 
${\bf f}({\bf s})$  is  in ${\cal V} _k$ if 
it can be written as the sum of eigenfunctions of $\D$  times constant
vectors, where each  of the eigenfunctions is the product of
$k$ or fewer ladder operators  $\mcl_{-j}, j=1,\ldots, n$ 
applied to the eigenfunction
$\Phi_0({\bf s})$.  
\end{definition}
The following lemma will be used in our perturbation analysis.
\begin{lemma}
If $h({\bf s}) \in {\cal V}_k$, then $g({\bf s}) = \langle {\bf a} , {\bf s} 
\rangle h({\bf s} )$ is in ${\cal V } _{k+1} $.
\label{lem:V}
\end{lemma}
\begin{proof}
This is almost a direct consequence of Lemmas \ref{lem_commute} and \ref{lem:alpha_beta}. 
 From Lemma \ref{lem:alpha_beta} we know that
$g(s)$ can be written as a  sum of terms involving $\mcl_{-j} h({\bf s})$
and $\mcl_j h({\bf s})$ where $j>0$.   By definition, each of the
terms $\mcl_{-j} h({\bf s})$ are in ${\cal V  }_{k+1}$.  On the other hand,
the commutator relations  $[ \mcl _j, \mcl _{-k}]   = - \delta_{jk} \Id$
from Lemma \ref{lem_commute}, and the fact that $\mcl_j 
\Phi_0({\bf s}) =0$ for $j>0$, can be used to show that $\mcl_j h({\bf s})$
is in ${\cal V }_{k-1}$  That is, $\mcl_j$  has either canceled  out
a previous term $\mcl_{-j}$ 
 applied to $\Phi_0$, or it  commutes with all of the previous operators applied
to $\Phi_0$, yielding the zero function because $\mcl_j \Phi_0 = 0$
for $j>0$.
\end{proof}

The perturbation analysis proceeds as follows. We have a zeroth-order solution
$\bfm_0 = \bfphi_1\Phi_0$, which is  clearly in ${\cal V}_0$.
We will see by induction, that the function
${\bf m} _k({\bf s})$  will be in ${\cal V}_k$.

 The equation at each higher  order will be 
of the form
\begin{equation}\label{eq:genk}
  (\lambda_0 -\D - \gzero )\bm_k = -\lambda_k \bm_0 + 
{\bf r}_{k-1}(s) 
\end{equation}
where ${\bf r} _{k-1}(s)$ is function that can be computed using
the  $\bm_j$ and  $\lambda_j$ for $j<k$.  
In particular, we have
\begin{equation*}
{\bf r} _{k-1}({\bf s} ) =  - \sum_{j=1}^{k-1} \lambda_j {\bf m} _{k-j}
+ \as \gone  {\bf m} _{k-1} 
\end{equation*}
Assuming  that  for $j<k$
 the functions ${\bf m} _j({\bf s} )$ are
 in ${\cal V}_j$, 
then  Lemma \ref{lem:V} ensures that 
the term ${\bf r} _{k-1}({\bf s})$  will
be  in ${\cal V}_k$.  
We can write
\begin{equation*}
{\bf r} _{k-1}({\bf s} ) = \Phi_0({\bf s})  {\bf b}_ {k-1}
 + \hat{{\bf r  }} _{k-1}({\bf s})
\end{equation*}
where the term $\hat{{\bf r }}_{k-1} ({\bf s})$  can be written as 
a sum of eigenfunctions of $\D$ times constant vectors, where none of the eigenfunctions is
$\Phi_0({\bf s})$.  
With this in mind we use Lemma \ref{lem:solve} to  see that 
we will be able to solve equation (\ref{eq:genk}) if and only if
$\lambda_k \ip{{\bfpsi} _1}{ {\bfphi} _1}= \ip{{\bfpsi} _1}{{\bf b} _{k-1}}$, and hence
\begin{equation*}
\lambda_k  =  \ip{ {\bfpsi} _1}{{\bf b} _{k-1}}.
\end{equation*}
Once we have chosen $\lambda_k$ in this way, we can 
solve for ${\bf m} _k$, and it will  clearly be in ${\cal V}_k$, thus
allowing us to continue the process to the next value of $k$  by induction.

 Terms in ${\bf r} _{k-1}({\bf s})$ 
proportional to $\Phi_0$ can only arise at even steps in the process (i.e. equations
for $\lambda_{2j}, \bfm_{2j}$) because $\mcl_k\Phi_k=-\Phi_0$ (see equation
\ref{eq:lem:phik}).
These terms proportional to $\Phi_0$ must satisfy the compatibility condition
$\ip{\bfpsi_1}{\bfb}=0$ as in Lemma \ref{lem:solve}.

\subsection{First Order}\label{sec:pert1}
To simplify  notation, we make the following definition.
\begin{definition}\label{def:phik}
The functions $\Phi_k({\bf s})$ are defined as
\begin{equation}\label{eq:phik}
 \Phi_k(\bfs)=\mcl_{-k}\Phi_0(\bfs), \qquad k = 1,\dots ,n,
\end{equation}
where $\Phi_0({\bf s})$ is the eigenfunction of $\D$  associated with the
eigenvalue with the largest real part.  Note that,
from Lemma \ref{lem_commute} and Theorem \ref{thm:lkphi}, we have
\begin{equation} 
\mcl_k \Phi_k(\bfs) = -\Phi_0(\bfs),\qquad k=1,\ldots, n,
\label{eq:lem:phik}
\end{equation}
because $\mcl_k \Phi_k = \mcl_k\mcl_{-k}\Phi_0=
 (\mcl_{-k}\mcl_{k}-1)\Phi_0=-\Phi_0$.
\end{definition}

Substituting \eq{expand} into \eq{general} and
collecting terms of order $\varepsilon$, 
we  get the equation for $\bm_1$
\begin{equation}\label{eq:eq1}
  (\lambda_0 -\D - \gzero )\bm_1 = -\lambda_1 \bm_0 + \as \gone \bm_0.
\end{equation}
It is not hard to show that the eigenvalue $\lambda(\ep) = \lambda_0+\ep \lambda_1 +\ldots$
must be an even function of $\ep$. This is also intuitive because the sign of $\ep$ 
plays no role in \eq{x}. Thus, it is no surprise that $\lambda_1 = 0$. 
\begin{lemma}\label{lem:m1}
If ${\bf H}$ and ${\bf B}$ satisfy the basic conditions (Def. \ref{def:bc}),
we have $\lambda_1 = 0$ and
\begin{equation}\label{eq:m1}
{\bf m} _1 = \sum_{k=1}^n \Phi_k({\bf s}) {\bf c} _k 
\end{equation}
where  $\Phi_k$ is defined in \eqref{eq:phik}, and
\begin{equation}\label{eq:ck}
  {\bf c} _k =  \sum_{j=1}^J \frac{\beta_k \ip{\bfpsi_j}{\gone \bfphi_1}}
{\nu_1 - \nu_j + \mu_k}\bfphi_j
\end{equation}
\end{lemma}
\begin{proof}
Using (\ref{eq:as}) and that $\mcl_k\Phi_0 = 0$, $\mcl_{-k}\Phi_0 = \Phi_k$, we have
\begin{equation} \label{eq:eq11}
\as \gone \bm_0 = \sum_{k=1}^n \beta_k \gone \bfphi_1 \Phi_k(\bs)
\end{equation} 
Thus, the right side of \eq{eq1} is a finite sum of terms proportional to $\Phi_0, \Phi_1,
\ldots, \Phi_n$, and each term can be treated separately. We now apply 
 Lemma \ref{lem:solve} to equation (\ref{eq:eq1}) using
equation (\ref{eq:eq11}). 
The only term proportional to $\Phi_0$ is $-\lambda_1 \bm_0$. But according
Lemma \ref{lem:solve} this means $\ip{\bfpsi_1}{\lambda_1 \bfphi_1}=0$. Hence, $\lambda_1 = 0$.
The expression in \eq{sol1} applied to the $\Phi_k$ terms
for $k>0$  gives the expression in \eq{m1}  .
\end{proof}

\subsection{Second Order}\label{sec:pert2}
Substituting \eq{expand} into \eq{general} and
collecting terms of order $\varepsilon^2$, 
we get the equation for $\bm_2$
\begin{equation}\label{eq:eq2}
  (\lambda_0 -\D - \gzero )\bm_2 = -\lambda_2 \bm_0 + \as \gone \bm_1.
\end{equation}
The situation here is similar to that for $\bfm_1$, except that
the terms proportional to $\Phi_0(\bfs)$ come from $\bfm_0$ as well as
terms of the form $\mcl_k\Phi_k(\bfs) = -\Phi_0(\bfs)$. 
\begin{lemma}\label{lem:quick}
If ${\bf H}$ and ${\bf B}$ satisfy the basic conditions (Def. \ref{def:bc}),
the compatibility condition for \eq{eq2} implies
  \begin{equation}\label{eq:lam2}
     \lambda_2 =-\sum_{k=1}^n\ip{\bfpsi_1}{\alpha_k \gone \bfc_k},
  \end{equation}
where ${\bf c} _k$ is defined as in \eq{ck}.
\end{lemma}
\begin{proof}
Lemma \ref{lem:m1} shows that
$\bfm_1 = \sum_{k=1}^n \Phi_k(\bs) \bfc_k $.  This fact, and an application of
the result in Lemma \ref{lem:alpha_beta} implies that 
\begin{equation}
  \as\gone\bfm_1 = \left(\sum_{l=1}^n \alpha_l \mcl_{l}+\beta_l\mcl_{-l}\right)
\gone \sum_{k=1}^n \Phi_k(\bs)\bfc_k  = - 
\Phi_0(\bs) \left(\sum_{k=1}^n\alpha_k \gone \bfc_k \right) + \ldots
\nonumber
\end{equation}
where the term on the right is the only term proportional to $\Phi_0$. 
We used \eqref{eq:lem:phik}   
to write $\mcl_k\Phi_k(\bfs)=\mcl_k\mcl_{-k}\Phi_0(\bfs) = -\Phi_0(\bfs)$.

Using the form of ${\bf m}_0$ in  \eqref{eq:m0},
the compatibility condition from Lemma \ref{lem:solve}
 implies $ -\lambda_2\ip{\bfpsi_1}{\bfphi_1} = \ip{\bfpsi_1}{\sum_{k=1}^n\alpha_k \gone \bfc_k }$. 
Hence
\begin{equation}
  \lambda_2  =-\ip{\bfpsi_1}{\sum_{k=1}^n\alpha_k \gone \bfc_k}=
-\sum_{k=1}^n\ip{\bfpsi_1}{\alpha_k \gone \bfc_k}.
\nonumber
\end{equation}
\end{proof}
Computing the expression for $\bfm_2$ is a simple exercise, but we do not 
write it here. Continuing this process for higher order terms is straightforward,
though grows more tedious with each successive order. 

 Lemma \ref{lem:quick}
allows us to compute $\lambda_2$, but 
a nice feature of the second order term  $\lambda_2$,
is that it can be expressed by a simple formula involving 
the extended power spectral density $G$ of 
the process $\ats$ (see Appendix C).
We prove the following theorem in  Appendix D.
\begin{theorem}\label{thm:m2}
If $\rp{\nu_1-\nu_j + \mu_k}>0$ 
for each $j=1,\ldots, J$ and $k=1,\ldots , n$, 
and if
 ${\bf H}$ and ${\bf B}$ satisfy the basic conditions (Def. \ref{def:bc}),
then  
  \begin{equation}\label{eq:lambda2}
  \lambda_2  = \sum_{j=1}^J \ip{\bfpsi_1}{\gone \bfphi_j}\ip{\bfpsi_j}{\gone \bfphi_1}
G(\nu_1 - \nu_j).
  \end{equation}
Here $G(z)$ is the extended  power spectral density of the forcing
term $\langle {\bf a}, {\bf s} \rangle$.
\end{theorem}
\begin{remark}
Note that the coefficients $\ip{\bfpsi_1}{\gone \bfphi_j}\ip{\bfpsi_j}{\gone \bfphi_1}$
and the differences $\nu_1 - \nu_j$ depend only on the differential equation for $\bfx$
(i.e. only on the matrices $\bfca_0$ and $\bfca_1$), and the function $G$ depends only
on the filter $\ats$ (i.e. on $\bfch, \bfcb, \bfa$). It would be interesting to investigate
whether the same form as in \eqref{eq:lambda2}
 would hold for \emph{any} asymptotically stationary filter.
That is, if the expression for $\lambda_2$ would be
a linear combination of values of $G$, where the coefficients depend only on the 
physical system, and the places where $G$ is evaluated are 
given by the eigenvalues of that system.
 \end{remark}

\section{Applications}\label{sec:app}

\subsection{Second Moments for the Mathieu Equation}\label{sec:mat}
We  can write the Mathieu equation (\ref{eq:Mathieu})
as in \eq{x} using 
 a two-dimensional  vector 
${\bf x}^T = ( x_1,x_2)$.
In this case the   matrices $\azero, \aone$ in equation \eq{x} are 
\begin{equation}\label{eq:MA0A1}
 \azero = \left( 
 \begin{array}{cc}
   0 & 1 \\ -\os & -\gamma
 \end{array}\right), \quad
\aone = \left( 
 \begin{array}{cc}
   0 & 0 \\ 1 & 0
 \end{array}\right).
\end{equation}
We will consider the stability for the second moments.
We define
\begin{equation}\label{eq:Mmjk}
  m_{jk}(\bs,t) = \int_{\R^2} x_j x_k P(x_1,x_2,\bs,t)dx_1 dx_2,
\end{equation}

In this case 
The Fokker  Planck equation (\ref{eq:fp})  can be written as
\begin{equation}
\partial_t P =  \D P - 
\pdone{}{x_1}  \left( x_2 P \right)
- \pdone{}{x_2} \left(\left( - \omega_0^2 x_1 - \gamma x_2 + 
\langle {\bf a} , {\bf s} \rangle x_1   \right)  P\right)
\label{eqn_mat}
\end{equation}

If we multiply equation (\ref{eqn_mat})   by $x_1^2$, and integrate over all values of
$x_1$ and $x_2$, after integrating by parts we get the equation
\begin{equation*}
\partial _t m_{11} = \D m_{11} + 2 m_{12}
\end{equation*}
Similarly multiplying equation (\ref{eqn_mat}) by $x_1 x_2$ and $x_2^2$,
integrating over all $x_1$ and $x_2$, and applying integration by parts, we
get the equations
\begin{equation*}
\partial _t m_{12} = \D m_{12}   - \omega_0^2 m_{11} - \gamma m_{12} 
+ m_{22} + \langle {\bf a},{\bf s} \rangle m_{11} 
\end{equation*}
and
\begin{equation*}
\partial _t m_{22} = \D m_{22}   - 2\omega_0^2 m_{12} - 2 \gamma m_{22} 
 + 2 \langle {\bf a},{\bf s} \rangle m_{12} .
\end{equation*}

If we let
 $\bm = (m_{11}, m_{12}, m_{22})^T $
this can be written 
in the form of equation (\ref{eq:gzero}) where 
the matrices $\gzero, \gone$ in \eq{gzero} are given by
\begin{equation}\label{eq:gamma}
  \gzero = \left(
  \begin{array}{ccc}
    0 & 2 & 0 \\ -\os & -\gamma & 1 \\ 0 & -2\os & -2 \gamma
  \end{array} \right), \quad
\gone = \left(
  \begin{array}{ccc}
    0 & 0 & 0 \\ 1 & 0 & 0 \\ 0 & 2 & 0
  \end{array} \right).
\end{equation}
After assuming temporal behavior of the form $e^{ \lambda t}$
 we arrive at the eigenvalue problem
\begin{equation} 
     \lambda \bm = \D \bm + \gzero \bm + \ep \as \gone \bm.
\end{equation}
for ${\bf m}({\bf s})$ and $\lambda$, which is the same as 
\eq{general}.
We  will now apply the
results of Theorem \ref{thm:m2} to this set of equations.  

In the case of second moments, the eigenvalues $\nu_j$ of $\gzero$
are given by sums of two eigenvalues of $\azero$. I.e., 
$\nu_j = \sigma_\ell +\sigma_m$ where $\sigma_k$ are
eigenvalues of $\azero$.  
In the case of the Mathieu equation, 
the eigenvalues of $\azero$ are $\sigma_1,\, \sigma_2$, where 
$\sigma_1 = \ob{\sigma}_2 = \frac{-\gamma + i\sqrt{\og}}{2}$. 
Hence, there are three choices of $\nu_1$, given by 
$\nu_1 = -\gamma,$ or $\nu_1 =-\gamma \pm i\sqrt{\og}$, since they all have the same
real part. 
In the case $\nu_1 = -\gamma$, we have 
$\ip{\bfpsi_1}{\gone \bfphi_1} = 0$, 
so the $G(0)$ term does not appear. We also have
 $\ip{\bfpsi_1}{\gone \bfphi_2}\ip{\bfpsi_2}{\gone \bfphi_1}=
\ip{\bfpsi_1}{\gone \bfphi_3}\ip{\bfpsi_3}{\gone \bfphi_1} = \frac{2}{\og}$.
Hence
\begin{equation}\label{eq:lam2mat}
\lambda_2 = \frac{2}{\og}\left(G(\nu_1-\nu_2) + G(\nu_1-\nu_3)\right)
=\frac{2}{\og}S\left(\sqrt{\og}\right),
\end{equation}
where $S(\omega)$ is the power spectral density of $\ats$.
This follows
because (without loss of generality, taking
$\nu_2 = -\gamma - i\sqrt{\og} = \overline{\nu}_3$) we have 
$\nu_1-\nu_2 = i\sqrt{\og} = -(\nu_1-\nu_3)$, and $G(i\omega) + G(-i\omega) = S(\omega)$
(see Appendix C).

If we take either $\nu_1=-\gamma \pm \sqrt{\og}$, then the
expressions for $\lambda_2$ are
\begin{align*}
&(+)\quad \lambda_2 = \frac{2}{\og}\left(G\left( i\sqrt{\og}\right) - 2G(0)\right)\\
&(-)\quad \lambda_2 = \frac{2}{\og}\left(G\left( -i\sqrt{\og}\right) - 2G(0)\right).
\end{align*}
Both cases have the same real part of $\lambda_2$
\begin{equation*}
  \rp{\lambda_2} = \frac{1}{\og}\left( S\left(\sqrt{\og}\right)-2S(0)\right),
\end{equation*}
which is less than the expression in \eq{lam2mat}. Hence, we have proved
\begin{theorem}\label{thm:lam2}
If ${\bf H}$ and ${\bf B}$ satisfy the basic conditions (Def. \ref{def:bc}),
then the second moments of the Mathieu equation \eq{Mathieu} become unstable when
$\lambda(\ep)>0$ where
  \begin{equation}
    \lambda (\ep) = -\gamma + \frac{2}{\og}S\left(\sqrt{\og}\right)\ep^2 + \ldots
  \end{equation}
\end{theorem}
\subsection{Comparing Moments for the Mathieu Equation}\label{sec:compare}
If we perform the same analysis as in \S \ref{sec:mat}, but
 for the first and third
marginal moment equations instead of the 
second marginal moment equation, we obtain results similar to Theorem \ref{thm:lam2}.
If we denote the largest eigenvalue of the $p$th moment operator
$\D + \gzero + \ep \as \gone$ as $\lambda^{(p)}$, then up to 
second order, we have
\begin{align}
&  \lambda^{(1)}(\ep) = \frac{-\gamma + i \sqrt{\og}}{2}  + \left(\frac{G(i\sqrt{\og\, })}{\og} 
-\frac{G(0)}{\og}\right)\ep^2 + \ldots \nonumber \\
&  \lambda^{(3)}(\ep) =\frac{-3\gamma + i \sqrt{\og}}{2} + \nonumber \\
& \qquad \qquad
\left( \frac{3G\left(-i\sqrt{\og}\right)}{\og}   + \frac{4G\left(i\sqrt{\og}\right)}{\og}-\frac{G(0)}{\og}
\right)\ep^2 + \ldots \nonumber
\end{align}
($\gone$ and $\gzero$ depend on $p$, but we do not make that explicit in our notation.)
It is only the real parts of the eigenvalues that factor into the stability. We have
\begin{align}\label{eq:l1}
&  \rp{\lambda^{(1)}(\ep)} = -\frac{\gamma}{2} +
\frac{1}{2(\og)}\left(S\left(\sqrt{\og\, }\right)- S(0)\right)\ep^2 +\ldots \\
\label{eq:l2}
 &  \rp{\lambda^{(2)}(\ep)} = -\gamma +  \frac{2}{\og}S\left(\sqrt{\og}\right)\ep^2 + \ldots  \\
\label{eq:l3}
&\rp{\lambda^{(3)}(\ep) } = -\frac{3}{2}\gamma +
\frac{1}{2(\og)}\left(7S\left(\sqrt{\og}\right) - S(0)\right)\ep^2 + \ldots
\end{align}
In \cite{Kampen}, there is a heuristic treatment of the first moments of $\bfx(t)$.
There, Van Kampen writes a series for $\bfx(t)$, which he truncates at the $\ep^2$ term
and then averages to get an expression for $\avg{\bfx(t)}$ up to order $\ep^2$. 
He then points out that this new series is the solution to an ODE, up to order $\ep^2$.
The stability of $\avg{\bfx(t)}$ is then analyzed in terms of this new ODE. His result
for the Mathieu equation matches ours up to order $\ep^2$ (although, he considers the case
$\gamma = 0$). Our result is a rigorous treatment, applies to higher moments,
and we can find the solution to any order in $\ep$. We stop at $\ep^2$ in this paper
only for convenience.

If we assume that the $\rp{\lambda^{(p)}}$ becomes positive while $\ep$ is small (so we 
neglect the $\ep^4$ terms and higher), then we can use \eq{l1}, \eq{l2}, and \eq{l3} to solve
$\rp{\lambda^{(p)}}=0$ for $p=1,2,3$. Then we find that the second moments will become
unstable before the first moments. If $S\left(\sqrt{\og}\right) > S(0)$, then 
the third moment will become unstable before the second moment. If $S\left(\sqrt{\og}\right) \leq S(0)$,
then the second moment becomes unstable before the third.

\subsection{Numerical Results} \label{sec:numerics}
In this subsection 
we discuss the computation  of   the eigenvalue
 that determines the stability of the Mathieu
equation (\ref{eq:Mathieu}), with $\azero, \aone$ from (\ref{eq:MA0A1}) and 
$\gzero, \gone$ from (\ref{eq:gamma}). 
We do not restrict ourselves to small values of
$\ep$.  
We carry out these calculations by converting the eigenvalue problem
to an infinite dimensional system of linear equations, and
truncating this system after a finite number of terms.  
Our procedure converges  rapidly as the number  of terms in our
expansion  is
increased.  

We limit ourselves to the case of 
a  second-order filter given by \eq{filter},
with  $\bfch$, $\bfcb$, and $\bfa$ given by
\begin{equation} \label{eq:filter2}
   \bH = \left(\begin{array}{cc} -\mu_1 & 0\\ \beta & -\mu_2 \end{array} \right), \quad
\bB= \left(\begin{array}{cc} 1 & 0\\ 0 & 0 \end{array} \right) , \quad
\ba = \left( \begin{array}{c} a_1 \\ a_2 \end{array} \right),
\end{equation}
where $\beta, a_1, a_2 \in \R$, $\beta \neq 0$, and $\mu_1, \mu_2 > 0$.
The vector of second marginal moments $\bfm (\bfs)$, given by 
(\ref{eq:Mmjk}), satisfies
\begin{align}\label{eq:num1}
  \lambda \bfm &= \D \bfm + \gzero \bfm + \ep \as \gone \bfm \nonumber \\
& = \frac{1}{2}\partial_{s_1}^2 \bfm + \mu_1 \partial_{s_1} (s_1 \bfm)
-\partial_{s_2}((\beta s_1 - \mu_2 s_2)\bfm)+ \gzero \bfm + \ep \as \gone \bfm .
\end{align}
If we multiply \eq{num1} by $s_2^j$ and integrate with respect to $ds_2$, then we get
\begin{align}\label{eq:num2}
  \lambda \bfm_j &= \frac{1}{2}\partial_{s_1}^2 \bfm_j + \mu_1 \partial_{s_1} (s_1 \bfm_j)
+ j\beta s_1 \bfm_{j-1} + \nonumber\\
& \qquad \qquad - j\mu_2 \bfm_j+ \gzero \bfm_j + 
\ep a_1 s_1 \gone \bfm_j + \ep a_2\gone \bfm_{j+1},
\end{align}
where
\begin{equation*}
  \bfm_j (s_1) = \int_\R s_2^j \bfm(s_1,s_2)ds_2.
\end{equation*}
This is an infinite set of equations for the marginals $\{\bfm_j(s_1)\}$.
Let $\varphi_k(s_1) = H_k(\sqrt{\mu_1}s_1)e^{-\mu_1 s_1^2}$, where $H_k$ is the $k$th Hermite
polynomial. We expand $\bfm_j$ in the basis $\varphi_k$ as 
\begin{equation*} 
\bfm_j (s_1) = \sum_k \bfc^k_j \varphi_k(s_1).
\end{equation*} 
The $\varphi_k$ are eigenfunctions of the differential operator in the
$s_1$ variable in \eq{num2}; explicitly
\begin{equation*}
  \frac{1}{2}\partial_{s_1}^2 \varphi_k + \mu_1 \partial_{s_1} (s_1
  \varphi_k) = -k\mu_1 \varphi_k \quad k \geq 0.
\end{equation*}
The Hermite polynomials satisfy the
recursion relation $H_{k+1}(y) = 2yH_{k}(y)-2kH_{k-1}(y)$, hence
\begin{equation*}
  s_1\varphi_k(s_1) = \frac{1}{2\sqrt{\mu_1}}(\varphi_{k+1}(s_1) + 2k\varphi_{k-1}(s_1))\quad k\geq 0.
\end{equation*}
Thus, \eq{num2} simplifies and becomes an equation for $\bfc_j^k$
\begin{align}\label{eq:num3}
  \lambda \bfc_j^k &= \left(\gzero - (k\mu_1 + j\mu_2)\bfci\right)\bfc_j^k
+\frac{j\beta}{2\sqrt{\mu_1}}\left(\bfc_{j-1}^{k-1} + 2(k+1)\bfc_{j-1}^{k+1}\right)\nonumber \\
&\qquad \qquad +\frac{\ep a_1}{2\sqrt{\mu_1}}\gone 
\left(\bfc_{j}^{k-1} + 2(k+1)\bfc_{j}^{k+1}\right)+\ep a_2 \gone \bfc_{j+1}^k
\end{align}
If we consider a finite number of moments ${\bf m}_j$ for $j\leq N_m$, and truncate the
expansion in $\varphi_k$ at $k\leq N_h$, then we get an approximation to the 
doubly infinite system \eq{num3}. This can be written as a matrix equation
\begin{equation}\label{eq:evaleq}
  \bfcl \bfz = \lambda \bfz
\end{equation}
 where $\bfcl$ is an $(N_mN_hJ)\times(N_mN_hJ)$ matrix. This eigenvalue
problem can be solved quickly on a computer.

\begin{table}[!h!b!t]
\begin{center} 
\begin{tabular}{|c|c|c|c|c|}
\hline
 & $\lambda(\ep)$ & $E_2$ &  $E_4$ \\
\hline
$\ep = 0.01$ &$-9.89\times 10^{-3}$ & $5.74\times 10^{-8}$ & 
$2.20\times 10^{-11}$  \\
\hline
$\ep = 0.05$ &$-7.20\times 10^{-3}$ & $3.62\times 10^{-5}$& $3.44\times 10^{-7}$ \\
\hline
$\ep = 0.10$ & $1.65\times 10^{-3}$ & $5.96\times 10^{-4}$& 
$2.21\times 10^{-5}$ \\
\hline
\end{tabular}
\caption{Values of the error in computing $\lambda(\ep)$ (for second moments)
for three values of $\ep$. 
$E_2$ is the error from the second order expansion, and $E_4$ is the
error from the fourth-order expansion.
Parameter values: $\mu_1 = 1.8, \mu_2 = 0.9, \beta = 1, \gamma = 0.01, 
\omega_0 = 0.5, a_1=1, a_2=0.9$, $N_m = 7,\, N_h = 5$.}
\end{center}
\label{table_one}
\end{table}

Table 1 shows 
the computed value of $\lambda(\ep)$ for second moments, 
which is the largest eigenvalue of $\bfcl$ in (\ref{eq:evaleq}).
That is, $\lambda(\ep)$ is the largest eigenvalue
for the Mathieu equation with filter \eq{filter2} (in this case the largest
eigenvalue is real). $E_2$ is the error from a second-order perturbation
expansion. That is,  $E_2(\ep) = | \lambda_0 + \lambda_2\ep^2
-\lambda(\ep)|$ with $\lambda_0 = -\gamma$ and $\lambda_2$ 
is given in equation (\ref{eq:lam2mat}).
$E_4$ is the fourth-order error, $E_4(\ep) = | \lambda_0 + 
\lambda_2\ep^2 +\lambda_4\ep^4-\lambda(\ep)|$, where 
$\lambda_4$ is computed by performing the perturbation
analysis to order four (the formula for $\lambda_4$ is not presented
here). 
The method converges rapidly; the values of
$\lambda(\ep)$ in the table were computed for $N_m = 7$ and $N_h = 5$.

\subsection{Alternative Representation of $\lambda_2$}\label{sec:tensor}
We present a formula for $\lambda_2$  that involves only
$\azero$ and $\aone$, avoiding construction of $\gzero, \gone$.
We do not present all of the details because the bookkeeping 
can be quite cumbersome
(an interested reader can find the details in \cite{SAND}),
 but we believe the formula for
$\lambda_2$ will be useful for applications.
 For instance,
if one wants to compute the perturbation coefficients on a computer,
it is easy to build an algorithm based on equation \eqref{eq:lam2tensor}
below, 
since one only needs to input the filter 
$(\bfch, \bfcb, \bfa)$ and the matrices $\azero$ and $\aone$.

The equation
for the second marginal moments can be written as
\begin{equation*}
   \partial_t \bM = \D \bM + \azero \bM +  \bM \azero^T + \ep  \as \left(\aone \bM
+ \bM \aone^T \right)
\end{equation*}
where $\bM$ is the $N\!\times \! N$ 
symmetric matrix with $\bM_{ij} =\int_{\R^N}x_ix_jP(\bs,\bx,t)d\bx$.
In this case one can solve an eigenvalue problem for
the stability where we have eigenvalues and eigenmatrices.  
Looking for solutions of the form
$\widetilde{\bfcm}(\bfs,t) = e^{\lambda t}\bfcm(\bfs)$, 
yields the eigenvalue problem for $\bfcm(\bfs)$
\begin{equation*}
     \lambda \bfcm = \D \bfcm +\azero \bM +  \bM \azero^T + \ep  \as \left(\aone \bM
+ \bM \aone^T \right).
\end{equation*}
The marginal moment tensor $\bfcm$ is symmetric ($M_{jk}=M_{kj}$), so we will
use a basis of symmetric tensors to express $\bfcm$, and in turn
reproduce the results of \S \ref{sec:pert}. The basis
that is simplest is given by 
the eigenmatrices  ${\bfce}_{jk}$ (and adjoints by ${\bfcf}_{jk}$
with inner product $\ip{\bfce}{\bfcf} = {\rm tr}(\bfce^T\bfcf)$) 
\begin{equation*}
{\bfce}_{jk} = \frac{1}{2}\left(\bfh_j\bfh_k^T + \bfh_k\bfh_j^T\right),
\quad {\bfcf}_{jk} = \frac{1}{2}\left(\bfg_j\bfg_k^T + \bfg_k\bfg_j^T\right),
\end{equation*}
where $\bfh_j$ are eigenvectors of $\azero$ with eigenvalues $\sigma_j$, 
and $\bfg_k$ are the normalized
adjoint eigenvectors of 
$\azero$,  $\ip{{\bf g} _j}{{\bf h} _k} = \delta_{jk}$. The
eigenvalues of the ${\bfce}_{jk}$ are sums of the $\sigma_i$;
$\azero {\bfce}_{jk} + {\bfce}_{jk} \azero^T = (\sigma_j +\sigma_k){\bfce}_{jk}$.

The analogous result to Lemma \ref{lem:solve} is straightforward to 
show, and following the steps in \S \ref{sec:pert} we arrive at
the following result (note that the eigenvalues $\nu_j$ of $\gzero$
from \S \ref{sec:pert} and \S \ref{sec:compare}
are sums $\sigma_\ell + \sigma _m$).
 
\begin{theorem}\label{thm:tensorm2}  
Let ${\bf H}$ and ${\bf B}$ satisfy the basic conditions (Def. \ref{def:bc}),
 and let $\{ {\bf h}_j \}_{j=1}^N$ form a complete set.
For $q,r$ fixed, if $\rp{\sigma_q+\sigma_r-\sigma_j-\sigma_k + \mu_\ell}>0$ 
for each $j,k=1,\ldots, N$ and $\ell =1,\ldots , n$, then  
the order-two coefficient in the 
expansion $\lambda = \lambda_0 + \lambda_2 \ep^2 +\ldots$, with
$\lambda_0 = \sigma_p + \sigma_r$, is given by
\begin{equation}\label{eq:lam2tensor} 
  \lambda_2=8\sum_{j,k=1}^N
\frac{C_{jkqr}C_{qrjk} }{1+\delta_{qr}}
G(\sigma_q + \sigma_r -\sigma_j-\sigma_k),
\end{equation}
where
\begin{equation}
  C_{jk\ell m} = \frac{1}{4}\left(\delta_{jm}\ip{\bfg_k}{\aone \bfh_\ell}+
\delta_{km}\ip{\bfg_j}{\aone \bfh_\ell}
+\delta_{j\ell}\ip{\bfg_k}{\aone \bfh_m}+\delta_{k\ell}\ip{\bfg_j}{\aone \bfh_m}\right),\nonumber
\end{equation}
and $\bfh_j$ are eigenvectors of $\azero$ with eigenvalues $\sigma_j$, 
and $\bfg_k$ are the normalized
adjoint eigenvectors of $\azero$,  $\ip{{\bf g} _j}{{\bf h} _k}
= \delta_{jk}$. 
\end{theorem}

\section{ Conclusions}

We  have carried  out a perturbation analysis
to characterize the moment 
stability 
 of  parametrically forced
linear equations, where  the  forcing 
is colored noise coming out of an Ornstein-Uhlenbeck process.
Our analysis applies to arbitrary linear  systems, and can in principle
be carried out
to any order. 
Our analysis depends  on characterizing
 the spectrum  of the vector Ornstein-Uhlenbeck process using
ladder  operators.  Though this spectrum  has been
characterized elsewhere  
\cite{Liberzon, Metafune, roy}, we believe the ladder operator approach 
has been shown to be 
useful in  carrying out our perturbation analysis.

\section*{Acknowledgements}
We would like to thank John Torczynski for  motivating and
finding funding  for this work.  We  also
thank Jim Ellison, Nawaf Bou Rabee, and Rich Field  for
several fruitful discussions concerning stochastic differential equations.

\section{Appendix A: Supplementary  Material for \S \ref{sec:lad} }\label{sec:appA}

In this appendix we give several lemmas 
 used in \S \ref{sec:lad},  as well as  supplying the proofs of 
several of the lemmas used in that section.

\begin{lemma} \label{Dsymm}
The operator $\D$ defined in equation \eqref{eq:defD} can be expressed as in
equation \eqref{eq:Dsum},
where the $d_{jk}$ are the components of 
the symmetric matrix $\bfcd$, given in 
equation \eqref{eq:AD}.
\end{lemma}
\begin{proof}
With $d_{jk}$ as the components of $\bfcd$ given in
equation \eqref{eq:AD}, we have
\begin{equation}
\frac{1}{2} \sum_{i=1}^{2n+1}\sum_{j=1}^{2n+1}  d_{ij}
L_i L_j =
\frac{1}{2} \sum_{i=1}^n \sum_{j=1}^n
\left(b_{ij} L_i L_j
- h_{ij} L_i L_{j+n} - h_{ji} L_{i+n} L_j  \right)
-  \frac{1}{2} \tr{\bfch} 
\label{eqn_temp}
\end{equation} 
The part of the operator involving the coefficients $b_{ij}$  is
clearly equal to the operator $\frac{1}{2}\Div{\bfcb\nabla \cdot}$.  
To show that
the left hand side of equation (\ref{eqn_temp}) is actually $\D$, we need to shows
that  the terms involving $h_{ij}$ are in fact the same
as $-\sum_{i=1}^n \sum_{j=1}^n h_{ij} L_i L_{j+n}=-\Div{\bfch \bfs \cdot}$.  
We compute
\begin{align*}
&\frac{1}{2} \sum_{i=1}^n \sum_{j=1}^n
\left( 
 h_{ij} L_i L_{j+n} + h_{ji} L_{i+n} L_j  \right)
+  \frac{1}{2} tr({\bf H}) \\
&=  \frac{1}{2} \sum_{i=1}^n \sum_{j=1}^n
\left( 
 h_{ij} L_i L_{j+n} + h_{ij} L_{j+n} L_i  \right)
+  \frac{1}{2} tr({\bf H}) \\
&=  \frac{1}{2} \sum_{i=1}^n \sum_{j=1}^n
\left( 
 h_{ij} L_i L_{j+n} + h_{ij} \left(L_i L_{j+n}  - \delta_{ij} \right)  \right)
+  \frac{1}{2} tr({\bf H})
=   \sum_{i=1}^n \sum_{j=1}^n 
 h_{ij} L_i L_{j+n}  
\end{align*}
In the second to last line above, we used the commutator
relation from (\ref{eqn_comL}).
\end{proof}

\begin{center}
{\bf Proof of Lemma \ref{matrixEig} }
\end{center}

\begin{proof}
We compute
an expression for $[\D,\mcl]$ in terms of $\bD$ and $\bA$.
\begin{align*}
  [\D,\mcl] & =  \sum_{i,j,m}\frac{1}{2}d_{ij}y_m (L_i L_j L_m - L_mL_iL_j) \\
& = \sum_{i,j,m}\frac{1}{2}d_{ij}y_m (L_i [L_j, L_m] + [L_i,L_m]L_j)\\
& = \sum_{i,j,m}\frac{1}{2}d_{ij}y_m (L_i a_{j,m} + a_{i,m}L_j)
=\sum_{i,m}\left(\frac{1}{2}\left(\bD+ \bD^T\right)\bA\right)_{i,m}y_mL_i.
\end{align*}
For the equation $[\D,\mcl]=\mu \mcl$, this implies that we have
\[
\sum_{i,m}\left(\frac{1}{2}(\bD+ \bD^T)\bA\right)_{i,m}y_mL_i 
= \mu \sum_i y_i L_i.
\]
In matrix notation, this is just $\bfcd\bA \bfy 
= \mu \bfy$, because $\bfcd =\bfcd^T$.
\end{proof}
This proof holds even if we do not assume that $\bD$ is symmetric.
In that case the analysis that follows would be done in terms
of the symmetric matrix $\bS = \frac{1}{2}(\bD + \bD^T)$,
instead of $\bD$. 
Thus, it is only for convenience that we use the symmetric form of $\bfcd$
in \eqref{eq:AD}.


\begin{center}
{\bf Proof of Lemma \ref{lem:T} }
\end{center}

\begin{proof}

 We denote the eigenvalues of ${\bf H}$  as $-\mu_k$  
with $\rp{\mu_k} >0$ for $k=1,2,\ldots,n$. 
Let $\bfu_k$ be the eigenvectors of $\bfch$ and $\bfv_k$ be the  adjoint eigenvectors
\begin{equation} 
{\bf H} {\bf u} _k = - \mu_k {\bf u} _k ,
\qquad
{\bf H}^T {\bf v} _k = - \overline{\mu}_k {\bf v} _k
\label{eqn_Hu}
\end{equation}
normalized so  that 
\begin{equation*}
\ip{\bfv_k}{\bfu_j}=\delta_{jk}.
\label{eqn_Huv}
\end{equation*} 
Recall that $\bfch$ is a real matrix, so complex eigenvalues
come in complex conjugate  pairs.  
If we write $\bfy = (\bfp, \bfq, r)^T$ then
$\bfct \bfy = \mu \bfy$ becomes
\begin{equation}\label{eq:T_eig}
\left( \begin{array}{ccc} \bH & \bB & 0\\ \,\, \bZero_n & -\bH^T & 0 \\
0 & 0& 0 \end{array} \right)
\left( \begin{array}{c} \bfp \\ \bfq \\ r \end{array} \right) = \mu
\left( \begin{array}{c} \bfp \\ \bfq \\ r \end{array} \right).
\end{equation}
There is a solution  with 
  $\mu = 0$  and  $\bfy^0 = (0,\ldots,0,1)^T$. If $\mu \neq 0$
then $r=0$, and
we have two cases. If $\bfq = \bZero$ then \eq{T_eig} reduces to $\bfch \bfp = \mu \bfp$.
Hence, $\mu = -\mu_k$ and $\bfp = \bfu_k$ for some $k$. We will denote this
solution as $\bfy^{-k} = (\bfu_k, \bZero, 0)^T$. If $\bfq \neq \bZero$ then we
must have $\bfch^T \bfq = -\mu \bfq$, so $\mu = \mu_k$ and $\bfq = \ob{\bfv}_k$ for some
$k$. We denote the solution in this case as
$\bfy^k = (-(\bfch -\mu_k \bfci)^{-1}\bfcb\ob{\bfv}_k, \ob{\bfv}_k , 0)^T$. 
\end{proof}
\begin{remark}\label{rem:mu}
$\bfct$ has the eigenvalue $0$, with 
corresponding ladder operator $\mcl_0=1$. This implies that $\mu_0 = 0$. However,
in this degenerate case, it is  convenient for notational purposes
to define $\mu_0 = -\tr{\bfch}$. We will also write $\mu_{-k}$ in
place of $-\mu_k$ to accommodate negative indices in the proof of Lemma \ref{lem:Dsum}.
\end{remark}

\begin{center}
{\bf The Proof of Lemma \ref{lem_commute} }
\end{center}

We denote by $\bfw^{\pm k}$ the normalized adjoint eigenvectors of $\bfct$. 
That is, $\bfct^T \bfw^{\pm k} = \pm\omu_k\bfw^{\pm k}$,
$\ip{\bfw^{\pm k}}{\bfy^{\pm k}}=\delta_{jk}$, and
$\ip{\bfw^{\pm k}}{\bfy^{\mp k}} = 0$. 
We begin with a preliminary lemma.


\begin{lemma}
\label{lem:EigT}
Let ${\bf u} _k$ and  ${\bf v} _k$ be the eigenvectors of
${\bf H}$ as in equations  (\ref{eqn_Hu}).  
Let
 ${\bf y}^{\pm k}$, $k=1,\ldots, n$, be  the eigenvectors of ${\bf T}$  associated with 
the eigenvalue 
$\pm \mu_{k}$, and let ${\bf w} ^{\pm k},k=1,\ldots, n$ be the normalized adjoint
eigenvectors. Then for each $k=1,\ldots, n$, 
$ {\bf A} {\bf y} ^{\pm k}  = {\bf \overline{w} } ^{\mp k}$. 
For $k=0,1,\ldots, n$,
$\bfcd {\bf w} ^{ \pm k} = \overline{\mu}_k \overline{{\bf y}}^{\mp
  k}$ (using $\mu_0 = -\tr{\bfch}$ from Remark \ref{rem:mu}).
Finally,  $\sum_{k=-n}^n \ob{w}^k_iy^k_j = \delta_{ij}$.
\end{lemma}
\begin{proof}
Note that $\bfy^{\pm k}$ are given explicitly in the proof of Lemma
\ref{lem:T} and for $k\neq 0$
\begin{equation}
  \label{eq:yk}
\bfy^{-k} = (\bfu_k, \bZero, 0)^T,\qquad 
\bfy^k = (-(\bfch -\mu_k \bfci)^{-1}\bfcb\ob{\bfv}_k, \ob{\bfv}_k , 0)^T.
\end{equation}
We define 
\begin{equation}
  \label{eq:wk}
\bfw^k = (\bZero, -\ob{\bfu}_k,0)^T,  \qquad
\bfw^{-k}=(\bfv_k, (\bfch-\omu_k \bfci)^{-1}\bfcb\bfv_k,0)^T
\end{equation}
and ${\bf y}^0 = {\bf w}^0 = (0,\ldots,0,1)^T$.
It is straightforward
to check that 
$\ip{\bfw^{\pm j}}{\bfy^{\pm k}}=\delta_{jk}$,
$\ip{\bfw^{\pm j}}{\bfy^{\mp k}} = 0$, $\bfct^T \bfw^0 ={\bf 0}$, 
and for $k\neq 0$, $\bfct^T \bfw^{\pm k} = \pm\omu_k\bfw^{\pm k}$,
so $\bfw^{\pm k}$ 
are the normalized adjoint eigenvectors.
Applying ${\bf A}$ to  the 
${\bf y}^{\pm k}$ in \eqref{eq:yk}
gives ${\bf A} {\bf y} ^{\pm k}  = {\bf \overline{w} } ^{\mp k}$ for
$k \neq 0$, and hence
applying $\bfcd$ to ${\bf A} {\bf y} ^{\pm k}  = {\bf \overline{w} } ^{\mp k}$
gives $\bfcd {\bf w} ^{ \pm k} = \overline{\mu}_k \overline{{\bf
    y}}^{\mp k}$ for $k \neq 0$.
With $\mu_0 = -\tr{\bfch}=\overline{\mu}_0 $ (since ${\bf H}$ is real), we have
$\bfcd {\bf w} ^{ 0} = \overline{\mu}_0 \overline{{\bf y}}^{0}$.
 (Note that ${\bf A} {\bf y} ^{0}  ={\bf 0}$, so without the convention
 in Remark \ref{rem:mu} we would not have 
$\bfcd {\bf w} ^{ 0} = \overline{\mu}_0 \overline{{\bf y}}^{0}$.)

We define the $(2n+1)\!\times \! (2n+1)$
matrices $\bfcy = [\bfy^{-n},\ldots, \bfy^n]$ and 
$\bfcw = [\bfw^{-n},\ldots, \bfw^n]$,
then $\bfcw^* \bfcy = \bI_{2n+1}$ because $(\ob{\bfw}^i)^T \bfy^j = \delta_{ij}$ for 
$-n \leq i,j \leq n$. But this means $\bfcy \bfcw^*= \bfci_{2n+1}$ as well, and the 
components of $\bfcy \bfcw^*$ are $(\bfcy \bfcw^*)_{ij} = \sum_{k=-n}^n \ob{w}^k_iy^k_j$.
\end{proof}

We now give the proof of 
Lemma \ref{lem_commute}.

\begin{proof}[Proof of Lemma \ref{lem_commute}]
Recall $\bA$ was defined as having coefficients $a_{mp}=[L_m,L_p]$.
Writing out $[\mcl_{\pm j},\mcl_k]$ in terms of the $L_m$ we have
  \begin{align*}
   [\mcl_{\pm j},\mcl_k] &= \sum_{m,p=1}^{2n+1} y^{\pm j}_my^k_p[L_m,L_p] 
 = \sum_{m,p=1}^{2n+1} y^{\pm j}_my^k_pa_{mp} \\
& = (\bfy^{\pm j})^T \bfca \bfy^k = \la \ob{\bfy}^{\pm j}, \bfca \bfy^k \ra.
  \end{align*}
Using $ {\bf A} {\bf y} ^{\pm k}  = {\bf \overline{w} } ^{\mp k}$
 we have 
$\la \ob{\bfy}^{\pm j}, \bfca \bfy^k \ra
=
\la \ob{\bfy}^{\pm j}, \overline{{\bf w}}^{-k}  \ra
= \overline{\la \bfy^{\pm j}, {\bf w}^{-k}  \ra}$.
Hence, $[\mcl_j,\mcl_k] =   \overline{ \ip{{\bf y} ^{+j}} {{\bf w } ^{-k}} }=0$
 and $[\mcl_{-j},\mcl_{k}] = 
   \overline{\ip{\bfy^{-j}}{\bfw^{-k}}}=\delta_{jk}$.
\end{proof}
\begin{center}
{\bf The Proof of Lemma \ref{lem:Dsum} }
\end{center}

\begin{proof}[Proof of Lemma \ref{lem:Dsum}]
We first consider $\frac{\mu_k}{2}\mcl_{-k}\mcl_{k}  = 
\sum_{p,m=1}^{2n+1}\frac{\mu_k}{2}y^{-k}_my^k_pL_mL_p$,
for each $k =-n, \ldots, n$, using the conventions in Remark
\ref{rem:mu}. 
For each $k$, $\bD \overline{{\bf w} }^k=\mu_k\bfy^{-k}$, which 
follows from Lemma \ref{lem:EigT}. Hence,
$y^{-k}_m = \frac{1}{\mu_k}\sum_{q=1}^{2n+1}d_{mq}\ob{w}^{k}_q$, so
if we replace the term $y^{-k}_m$ in the above expression
for $\frac{\mu_k}{2}\mcl_{-k}\mcl_{k}$, and sum over $k$, we get 
\begin{align*}
  \sum_{k=-n}^n \frac{\mu_k}{2}\mcl_{-k}\mcl_{k}
 & =  \sum_{k=-n}^n \sum_{p,m,q=1}^{2n+1}\frac{\mu_k}{2}\frac{1}{\mu_k}d_{mq}
\ob{w}_q^k y^k_pL_mL_p \\
& =  \sum_{p,m,q=1}^{2n+1} \frac{1}{2}d_{mq}L_mL_p \sum_{k=-n}^n \ob{w}_q^k y^k_p.
\end{align*}
From Lemma \ref{lem:EigT}, $\sum_{k=-n}^n \ob{w}_q^k y^k_p = \delta_{qp}$, so
\begin{equation}\label{eq:D1}
 \sum_{k=-n}^n \frac{\mu_k}{2}\mcl_{-k}\mcl_{k}= \sum_{p,m,q=1}^{2n+1} 
\frac{1}{2}d_{mq} \delta_{qp} L_mL_p 
= \sum_{p,m=1}^{2n+1}\frac{1}{2}d_{mp}L_mL_p = \D.
\end{equation}
For each $k > 0$,
 we can write $\frac{\mu_{k}}{2}\mcl_{k}\mcl_{-k} = \frac{\mu_{k}}{2}
\left( \mcl_{-k}\mcl_{k} -1\right)$ by the result of
 Lemma \ref{lem_commute}. Combining this with \eq{D1} and using $\mu_0 = -\tr{\bfch}$
we can write $\D$ as 
\begin{equation*}
  \D = -\frac{1}{2}\tr{\bfch} + \sum_{k=1}^n \left\{
 \frac{\mu_k}{2}\mcl_{-k}\mcl_{k} +\frac{\mu_{k}}{2}
\left( \mcl_{-k}\mcl_{k} -1\right) \right\}.
\end{equation*}
But the eigenvalues of $\bfch$ are $-\mu_k$, hence $\tr{\bfch} = -\sum_{k=1}^n\mu_k$ and 
we have $\D = \sum_{k=1}^n\mu_k\mcl_{-k}\mcl_{k}$.
\end{proof}

\section{Appendix B:  Supplementary Material for \S \ref{sec:eig}}
\begin{center}
{\bf Proof of Lemma  \ref{lem:boundD} }
\end{center}
\begin{proof}
  Suppose $\chi$ is an eigenvalue of $\D$ with eigenfunction $\phi$,
$\int_{\R^n}|\phi|^2d\bfs = 1$.
If we multiply \eq{eig_prob} by $\overline{\phi}$,  use the
definition of $\D$ in \eqref{eq:defD}, integrate over all
of space, and integrate the term involving ${\bf B}$ by parts, we  get
\begin{align}
 \chi \int_{\R^n}|\phi|^2d\bfs  = \int_{\R^n}\ob{\phi}\D \phi d\bfs &=  \int_{\R^n}
\ob{\phi}\frac{1}{2}\Div{\bfcb \nabla \phi}-\ob{\phi}\Div{\bfch \bfs \phi} d\bfs \nonumber \\
& =  -\int_{\R^n}\frac{1}{2}\ip{\nabla \phi}{\bfcb \nabla \phi}+
\ob{\phi}\Div{\bfch \bfs \phi} d\bfs \nonumber
\end{align}
The matrix $\bfcb$ is positive semi-definite, so $\ip{\nabla \phi}{\bfcb \nabla \phi}
\geq 0$, hence $\rp{\chi} \leq \rp{-\int_{\R^n}\ob{\phi}\Div{\bfch \bfs \phi} d\bfs}$.
But, because $\bfch$ is real,
\begin{align}
  2\rp{\int_{\R^n}\ob{\phi}\Div{\bfch \bfs \phi} d\bfs} &= 
\int_{\R^n}\ob{\phi}\Div{\bfch \bfs \phi}+\phi\Div{\bfch \bfs \ob{\phi}} d\bfs. \nonumber 
\end{align}
If we integrate the first  term  on the right in this expression by parts, and
expand the second  term we get
\begin{align}
  2\rp{\int_{\R^n}\ob{\phi}\Div{\bfch \bfs \phi} d\bfs} &= 
  \int_{\R^n}-\nabla \ob{\phi}\cdot (\phi \bfch\bfs) + 
\phi(\ob{\phi}\tr{\bfch} + (\bfch \bfs) \cdot  \nabla\ob{\phi})d\bfs \nonumber \\
& = \int_{\R^n}|\phi|^2\tr{\bfch}d\bfs = \tr{\bfch}.\nonumber
\end{align}
Hence, $\rp{\chi}\leq -\frac{1}{2}\tr{\bfch}$.
\end{proof}
\begin{center}
{\bf Proof of Lemma \ref{lem:Sigma} }
\end{center}

The proof of Lemma  \ref{lem:Sigma} follows almost 
immediately  from a few preliminary lemmas.

\begin{lemma}\label{lem:SigmaInverse}
 Suppose the eigenvectors
${\bf q}_k$ of ${\bf H}$ are complete and the adjoint eigenvectors
${\bf p}_k$
are normalized so $\ip{{\bf p}_j}{{\bf q}_k}=\delta_{jk}$.
Let ${\bf P} = [{\bf p} _1,{\bf p} _2,\ldots,{\bf p} _n]$,
${\bf Q} = [{\bf q} _1,{\bf q} _2,\ldots,{\bf q} _n]$.
We have
\begin{equation}
{\bf H} {\bf \Sigma} ^{-1} + {\bf \Sigma} ^{-1} {\bf H} ^T = - {\bf B}
\label{eqn_siginv}
\end{equation}
where ${\bf \Sigma}^{-1} = {\bf P} {\bf Q} ^{-1}$.
\end{lemma}
\begin{remark}
 Lemmas \ref{lem_carlson} and \ref{lem:Sigma}
show that, with the appropriate assumptions on $\bfch$ and $\bfcb$,
the matrix $\bfcp$ is invertible
and thus there exists a nonsingular matrix $\bfS = \bfcq\bfcp^{-1}$, so our use of the notation
$\bfS^{-1}$ is appropriate.
\end{remark}
\begin{proof}
According to equation (\ref{eq:T_eig}) we have
${\bf H} {\bf p} _k + {\bf B} {\bf q} _k = \mu_k {\bf p} _k$, and
$- {\bf H} ^T {\bf q} = \mu_k {\bf q} _k$.  
Writing this out in matrix form we get
${\bf H} {\bf P} + {\bf B} {\bf Q} = {\bf P} {\bf M}$, 
$- {\bf H} ^T {\bf Q} = {\bf Q} {\bf M}$.
Here ${\bf M}$ is the diagonal matrix with $\mu_k$ on the $k$th 
diagonal.  Using the second of these equations to
write ${\bf M}$ in terms of ${\bf Q}$ and ${\bf H}$,
and assuming ${\bf Q}$ is  invertible (the eigenvectors of
${\bf H}$ are complete)  we get
${\bf M} = - {\bf Q}^{-1} {\bf H} ^T {\bf Q} $.  
Substituting this into the  first equation we get
$ {\bf H} {\bf P} + {\bf B} {\bf Q} = - {\bf P} {\bf Q} ^{-1} {\bf H} ^T {\bf Q}$.
If we multiply this by ${\bf Q} ^{-1}$ on the right and rearrange,
we get the result of the
lemma.
\end{proof}

We will use the following result for  controllable pairs,
which follows immediately from Theorem 2 in \cite{carlson}.
\begin{lemma}
\label{lem_carlson}
If ${\bf B}$ is positive semi-definite, and the eigenvalues of
${\bf H}$ all have real parts less than zero, then
the solution to 
${\bf H} {\bf R} + {\bf R}  {\bf H} ^T = - {\bf B}$ is symmetric and
positive definite  provided $({\bf H}, {\bf B})$
form a controllable pair.  
\end{lemma}

Lemma \ref{lem:Sigma} 
 follows almost immediately from  the previous two lemmas.

\begin{center}
{\bf Some lemmas used in the proof of Theorem \ref{thm:integ} }
\end{center}

\begin{lemma}\label{lem_commute_pow}
For any integer $m \geq 0$, 
the  operators $\mcl_k$ and $\mcl_{-k}$ satisfy
\begin{equation}
\left[ \mcl_{-k}^{m+1}, \mcl_k \right] = 
(m+1) \mcl_{-k}^m
\end{equation}
\end{lemma}
\begin{proof}
For $m=0$, this follows immediately from  Lemma
\ref{lem_commute}. 
We can now proceed by induction.  In particular, if 
$\mcl_{-k} ^m \mcl _k - \mcl _k \mcl_{-k}^m  = m \mcl_{-k} ^{m-1}
$, then  if we multiply both sides of this equation by $\mcl_{-k}$ and
use $\mcl_{-k}^m \mcl_k \mcl_{-k}=
\mcl_{-k}^m \left( -I + \mcl_{-k} \mcl_k \right) =
\mcl_{-k}^{m+1} \mcl_k   - \mcl_{-k}^m$, we  find that
$\mcl_{-k}^{m+1} \mcl_k - \mcl_k \mcl_{-k}^{m+1} = (m+1)
\mcl_{-k}^m$, which proves the lemma.
\end{proof}
\begin{lemma}\label{lem_efuncs}
Let ${\bf H}$ and ${\bf B}$ satisfy the basic conditions (Def. \ref{def:bc}),
and let ${\bf k} = (k_1,k_2,..k_n)$ be a vector of nonnegative integers.
Let
\begin{equation}
\Phi_{{\bf k}}({\bf s}) = \mcl_{-1}^{k_1}
\mcl_{-2}^{k_2}....\mcl_{-n}^{k_n}
\Phi_0({\bf s}),
\end{equation}
then $\Phi_{{\bf k}}( {\bf s}) $ is nonzero, and has an eigenvalue of
\begin{equation}\label{eq:chi_k}
\chi_{{\bf k}} = - \sum_{j=1}^n k_j \mu_j 
\end{equation}
\end{lemma}
\begin{proof}
 We begin by showing that $\mcl_{-k}^m \Phi_0({\bf s})$ is
nonzero for all $m\geq 0$.  This clearly holds for $m=0$ by Lemma \ref{lem:phi0}.  
By induction we can see that if it is nonzero for $m-1$, then it is non-zero for $m$.  
This follows from the fact that 
$\left[ \mcl_{-k}^m,\mcl_k \right] =  m  \mcl_{-k}^{m-1}$, and  the
fact that $\mcl_k \Phi_0 =0$.  Combining these two facts we get
$- \mcl_{k} \mcl_{-k}^m  \Phi_0({\bf s})  = m  \mcl_{-k}^{m-1} \Phi_0({\bf s})$. 
This shows that  if $\mcl_{-k}^m$ vanished,then $\mcl_{-k}^{m-1}$ would
also have to vanish.   
Since we are assuming this is not the case, 
it follows that $\mcl_{-k}^m \Phi_0({\bf s} )$ does not vanish,  and
hence by induction
does not vanish for any $m\geq  0$.  

To show that a general function $\Phi_{{\bf k}}({\bf s})$ does not
vanish, we can proceed by a different induction proof.    In
particular, since the   operator $\mcl_{-1}$ commutes with both 
$\mcl_{-2}$ and $\mcl_2$ we see that for any operator 
$Z$ of the form $Z= \mcl_{-1}^p$ where $p$ is a non-negative integer, we have
$\left[Z \mcl_{-2}^m , \mcl_{2} \right] = m Z\mcl_{-2}^{m-1}  $.  
We can now use almost the identical argument as in the last paragraph to 
show that any function of the form $Z \mcl_{-2}^m \Phi_0$ will be non-zero.
We can now carry out this process by induction to see that
any function of the form $\Phi_{{\bf k} }({\bf s})$ will be nonzero.  

Once we know that $\Phi_{{\bf k} }( {\bf s})$ is nonzero, it  is clear from
the ladder operator formalism that its eigenvalue must  have the
form in \eqref{eq:chi_k}.
\end{proof}


There is one subtle point we would like to discuss in 
our proof of Theorem \ref{thm:integ}.  
Our proof relies on the fact that   if $\phi$ is an eigenfunction of
$\D$, then  either $\mcl _k \phi=0$, or $\mcl_k \phi$ gives a new
eigenfunction whose eigenvalue has a smaller real part.  
This relies on the assumption that $\mcl_k \phi$ remains in the domain
of our operator.  
The domain of our operator consists of functions that have moments of
all orders.  Clearly, if this is true of $\phi$, this will be true of
$\mcl _k \phi$.  However, we must also make sure that the function
$\mcl _k \phi$ has sufficient numbers of derivatives to satisfy 
our differential equation.   
This is clearly true of the eigenfunctions we have found.  That is,
they clearly have infinitely many derivatives.  
However, we should consider the possibility that there are other
eigenfunctions  that we have not accounted for  that 
are not infinitely differentiable.   
General theorems on elliptic operators rule out such eigenfunctions if
${\bf B}$ is positive definite.  However, we have only required that
${\bf B}$ be positive semi-definite, and that ${\bf H}$ and ${\bf B}$
form a controllable pair. 
A heuristic argument that we have found all of the eigenfunctions 
in this less restrictive case  is as follows.  
If we perturb the matrix ${\bf B}$  to make it positive definite, then
we know we have all of the eigenfunctions.  As our perturbation parameter
goes to zero, there is nothing unusual happening to our spectrum (such
as eigenvalues going off to infinity, or clustering about a point).  
Hence, if the eigenfunctions are complete for positive definite ${\bf B}$ 
they are clearly complete in the less restrictive case
where ${\bf B}$ and ${\bf H}$ form a controllable pair.


\section{Appendix C}\label{sec:appB}
In this appendix, we provide formulas for the 
asymptotic autocorrelation function of 
the process $\bfs (t)$ and the 
extended power spectral density
(defined in (\ref{eq:Gdef})) for $\bfs (t)$ as well
as for the filter $\ats$. In
particular, the results of Theorem \ref{thm:G} and  Corollary \ref{cor:S},
are used to 
express $\lambda_2$
in Theorems \ref{thm:m2}, \ref{thm:lam2}, and \ref{thm:tensorm2},
and throughout \S \ref{sec:app}.  Corollary \ref{cor:S} gives a practical formula
for computing the power spectral densities of $\bfs (t)$ and $\ats$.

\subsection{The Asymptotic Autocorrelation Function} \label{sec:auto}
We begin by  proving   a lemma concerning the autocorrelation function
of ${\bf s}(t)$ as defined in equation \eqref{eq:filter}. $\bfs(t)$ is not
a stationary process, but as $t \to \infty$ it approaches a 
stationary process, which we refer to as \emph{asymptotically stationary}.

\begin{lemma}
\label{lem_app2}
Suppose 
${\bf H}$ and ${\bf B}$ satisfy the basic conditions (Def. \ref{def:bc}),
and
let ${\bf s}(t)$ be the solution to equation \eqref{eq:filter} with 
zero initial conditions. 
As $ t \rightarrow \infty$ 
the autocorrelation function  ${\bf R}(\tau) = 
\avg{{\bf s}(t) {\bf s}^T(t+\tau)} $ is given by 
\begin{equation}
\label{eqn_Rapp}
{\bf R}(\tau) = {\bf \Sigma}^{-1} e^{  \bH ^T \tau }
\;\;\;\;\mbox{ for $\tau>0$},
\end{equation}
and ${\bf \Sigma}^{-1} = \bfcp \bfcq^{-1}$
satisfies equation (\ref{eqn_siginv}).
\end{lemma}
\begin{proof}
We define
\begin{equation}
\label{eqn_defK}
{\bf K}(t) = e^{ \bH t } {\bf B} e^{ \bH^T t}, \qquad 
{\bf K}_0 =  \lim_{t \rightarrow \infty}
\int_0^t {\bf K}(t-\sigma ) d\sigma.
\end{equation}
The solution  to equation \eqref{eq:filter}  (with zero initial conditions) 
 is given by
\begin{equation*}
{\bf s}(t) = \int_0^t e^{ \bH (t-s)} {\boldsymbol \xi}(s) ds
\end{equation*}
We can write
\begin{equation*}
{\bf s}(t) {\bf s}^T(t+ \tau)
= \int_0^t \int_0^{t+\tau} e^{  \bH (t-s) }
{\boldsymbol \xi}(s) {\boldsymbol \xi}^T(r) e^{  {\bf H^T} (t+\tau-r) }
 dr ds.
\end{equation*}
If we take the expected value of both sides of this equation,
and use the fact that 
$\avg{ {\boldsymbol \xi}(s) {\boldsymbol \xi}^T(r) } =
{\bf B} \delta(r-s),$
we arrive at the equation
\begin{equation}
\avg{ {\bf s}(t) {\bf s}^T(t+\tau) }
= \int_0^t  e^{  \bH (t-\sigma) }
{\bf B} e^{  \bH^T (t - \sigma)} e^{  \bH^T \tau } d\sigma =
\int_0^t{\bf K}(t-\sigma) d\sigma e^{  \bH^T \tau }.
\label{eqn_temp2}
\end{equation}
When deriving equation (\ref{eqn_temp2}) we have assumed that the 
variable $r$ is equal to the variable $s$ at some point when doing
the integration.  This will only be guaranteed   if $\tau>0$, and
hence this is only valid for $\tau>0$.  The expression for
$\tau<0$, is obtained by using the fact that the autocorrelation
function must satisfy ${\bf R} (-\tau) = {\bf R}^T  ( \tau)$.  

Assuming that all of the eigenvalues of $\bH$ have negative real part, 
the process ${\bf s}(t)$ will become stationary as $t \rightarrow \infty$.  
We take the limit of equation (\ref{eqn_temp2})  as $t \rightarrow
\infty$ to get
\begin{equation*}
{\bf R} (\tau) = {\bf K}_0 e^{  \bH ^T  \tau   }, 
\end{equation*}
where ${\bf K} _0$ is defined in equation (\ref{eqn_defK}). 
We now show ${\bf K}_0={\bf \Sigma}^{-1}$ by showing
${\bf K}_0$ satisfies equation (\ref{eqn_siginv}), i.e.
$\bH {\bf K}_0 + {\bf K}_0 \bH^T = -{\bf B} $.

We have from \eqref{eqn_defK}
\begin{equation*}
\odone{}{s}  {\bf K}(s) 
= \bH  {\bf K}(s) + {\bf K}(s) \bH ^T .
\end{equation*}
It follows that 
\begin{equation*}
\bH {\bf K}_0 + {\bf K}_0 \bH^T =
- \lim_{t \rightarrow \infty }  
\int_0^t \odone{}{s} \left( {\bf K} (t-s) \right) ds.
\end{equation*}
We can evaluate this integral using the fundamental theorem of  calculus.
When we do this we find that the contribution at $s=0$ vanishes in the limit
as $t \rightarrow \infty$.  Since ${\bf K}(0) = {\bf B}$, the
contribution at $s=t$ is just $-{\bf B}$, which completes the proof of
the lemma.
\end{proof}

\subsection{The Extended Power Spectral Density}
The expression for the eigenvalue (with largest real part)
of the perturbed operator
$\D +\gzero + \ep \as \gone$ will be written in terms of 
the Laplace transform of the asymptotic autocorrelation function
of the
asymptotically stationary filter $\ats$, which we denote by $G$. 
$G$ can be viewed as an extension of the power spectral density,
and has the advantage that
it can be evaluated at points in the complex plane,
outside of the domain of the power spectral density.
\begin{definition}
Let $\bfs(t)$ be an asymptotically stationary stochastic process 
(i.e. stationary in the limit
$t\to\infty$) with asymptotic autocorrelation function $\bfcr(\tau)$.  
We define the \emph{extended power spectral density} of $\bfs(t)$ as
  \begin{equation}
    \label{eq:Gdef}
\bfcg (z) = \int_0^\infty \bfcr (\tau) e^{-z\tau}\, d\tau.
  \end{equation}
With this definition, the scalar filter $\ats$ has extended power spectral density 
$G(z) = \ip{\bfa}{\bfcg (z) \bfa}$. $\bfcg$ is indeed an extension of the 
power spectral density $\bfcs(\omega) = \int_\R \bfcr(\tau)e^{-i\omega \tau}d\tau$,
because the domain of $\bfcg$ contains the set $\{z\in \C : \rp{z}\geq 0 \}$.
In particular, $\rp{\bfcg (i\omega)} = \frac{1}{2}\bfcs(\omega)$, which follows from
$\bfcr^T(\tau) = \bfcr(-\tau)$.
\end{definition}
\begin{theorem}\label{thm:G}
If
${\bf H}$ and ${\bf B}$ satisfy the basic conditions (Def. \ref{def:bc}), 
then
the extended power spectral density ${\bf G}(z)$ for the 
asymptotically stationary process ${\bfs} (t)$, 
defined in \eqref{eq:filter}, is given by 
\begin{equation}
  \label{eq:Gepsd}
  {\bf G}(z) = -{\bf \Sigma}^{-1}\left( \bfch^T - z \bfci \right)^{-1},
\end{equation}
provided  $\rp{\mu_{l}+z} >0$ for $l=1,\ldots , n$.

Furthermore,
the extended power spectral density $G(z)$ for the asymptotically stationary
filter $\ats$ can be written as
  \begin{equation}
    \label{eq:G}
G(z) = \ip{\bfa}{\bfcg (z) \bfa} = -\sum_{l=1}^n \frac{\alpha_l \beta_l}{\mu_{l}+z}, 
  \end{equation}
where $\alpha_l, \, \beta_l$ are defined in \eq{coeffs1}.
\end{theorem}
\begin{proof}
In Lemma \ref{lem_app2}, we showed that the
autocorrelation function of 
the asymptotically stationary
process $\bfs(t)$, in the limit $t\to \infty$, is given by
$\bfcr (\tau) = {\bf \Sigma}^{-1} e^{\bfch^T \tau}$ where
\begin{equation}\label{eq:r0}
{\bf \Sigma}^{-1} = \lim_{t\to \infty}\int_0^t e^{\bfch (t-s)}\, \bfcb\, e^{\bfch^T (t-s)} \, ds.
\end{equation}
From $\bfcr (\tau) = {\bf \Sigma}^{-1} e^{\bfch^T \tau}$, we have
\[
\int_0^\infty \bfcr(t)e^{-zt} \, dt = -{\bf \Sigma}^{-1}\left( \bfch^T - z \bfci \right)^{-1},
\]
assuming that $\rp{\mu_{l}+z} >0$ for $l=1,\ldots , n$ so that the integral converges.

Since $\bfa$ is real, we can use \eqref{eq:coeffs2} to write
$\bfa = \sum_{k=1}^n \alpha_k \ob{\bfv}_k = \sum_{k=1}^n \ob{\alpha}_k\bfv_k$.
Recall, we defined $\bfv_l$ so that $\bfch^T \bfv_l = -\ob{\mu}_l \bfv_l$, so we have
$( \bfch^T - z\bfci )^{-1} \ob{\bfv}_l = \frac{-1}{\mu_{l}+z}\ob{\bfv}_l$ 
and 
$e^{\bfch^T (t-s)}\ob{\bfv}_l = e^{-\mu_{l}(t-s)}\ob{\bfv}_l$ .
Using these expressions along with \eqref{eq:coeffs2}, \eqref{eq:r0},
and $\bfcb = \bfcb^T$, we compute
\begin{align}
  G(z) &= -\lim_{t\to \infty}\int_0^t\sum_{l,m=1}^n\ob{\alpha}_m \alpha_l
\bfv_m^T e^{\bfch (t-s)}\, \bfcb\, e^{\bfch^T (t-s)} ( \bfch^T - z\bfci )^{-1} \ob{\bfv}_l \, 
ds \nonumber \\
& = \sum_{l,m=1}^n\frac{\ob{\alpha}_m \alpha_l}{\mu_{l} + z}\ip{\ob{\bfv}_m}{\bfcb \ob{\bfv}_l} 
\lim_{t\to \infty}\int_0^t e^{-(\ob{\mu}_{m} +\mu_{l})(t-s)}ds\nonumber \\
& = \sum_{l,m=1}^n\frac{\ob{\alpha}_m \alpha_l}{(\ob{\mu}_{m} + \mu_{l})(\mu_{l} + z)}
\ip{\bfv_l}{\bB \bfv_m} 
=-\sum_{l=1}^n \frac{\alpha_l \beta_l}{\mu_{l}+z} .
\end{align}
\end{proof}

\begin{corollary}
\label{cor:S}
If 
${\bf H}$ and ${\bf B}$ satisfy the basic conditions (Def. \ref{def:bc}),
then
the  power spectral density $S(\omega)$ of the 
asymptotically stationary filter
$\ats$ is given by $S(\omega) = \ip{\bfa}{  \bfcs (\omega) \bfa}$, where
${\bf S}(\omega)$ is the  power spectral density of the 
asymptotically stationary process 
${\bfs} (t)$, defined in \eqref{eq:filter}, and
\begin{equation}\label{eq:Spsd}
\bfcs (\omega) =   \left(  \bH ^T + i \omega {\bf I} \right)^{-1}
\bfcb \left(  \bH ^T - i \omega {\bf I} \right)^{-1}.
\end{equation}
\end{corollary}
\begin{proof}
Using the expression for ${\bf G}$ in equation (\ref{eq:Gepsd}) we get
\begin{align*}
{\bf S}(\omega) &= 2\rp{{\bf G}(i\omega) } = {\bf G}(i\omega)+{\bf
  G}(i\omega)^* \\
& = -{\bf \Sigma}^{-1}\left( \bH^T - i \omega {\bf I} \right)^{-1} - 
\left( \bH + i \omega {\bf I} \right)^{-1} {\bf \Sigma}^{-1} \\
& =  - \left( \bH + i \omega {\bf I} \right) ^{-1}
\left( (\bH + i \omega {\bf I} ) {\bf \Sigma}^{-1} +
{\bf \Sigma}^{-1}  (\bH^T  - i \omega {\bf I}  ) \right)
\left( \bH^T - i \omega {\bf I} \right)^{-1} \\
& = - \left( \bH + i \omega {\bf I} \right) ^{-1}
\left( \bH  {\bf \Sigma}^{-1}  +
{\bf \Sigma}^{-1}  \bH^T  \right)
\left( \bH^T - i \omega {\bf I} \right)^{-1}\\
&=  \left( \bH + i \omega {\bf I} \right) ^{-1}
{\bf B}
\left( \bH^T - i \omega {\bf I} \right)^{-1}.
\end{align*}
The asymptotic autocorrelation function for $\ats = \bfa^T\bfs (t)$ is
given by
$\avg{\bfa^T\bfs(t) \bfs^T(t+\tau)\bfa}=\ip{\bfa}{ \bfcr(\tau) \bfa}$.
Hence $S (\omega) = \ip{\bfa}{ \bfcs(\omega) \bfa}$.
\end{proof}

\section{Appendix D: Supplementary Material for \S \ref{sec:pert}}

\begin{center}
{\bf Proof of Lemma \ref{lem:alpha_beta} }
\end{center}

\begin{proof}

We begin by defining
\begin{equation}\label{eq:coeffs1}
\alpha_k = \left(\bfcu^T\bfa\right)_k, \quad 
\beta _k = -\sum_{m=1}^n\frac{\ob{\alpha}_m}{\ob{\mu}_{m} + \mu_{k}}\ip{\bfv_k}{\bfcb \bfv_m}
\end{equation}  
where $\bfcu = [\bfu_1 , \bfu_2, \ldots , \bfu_n]$. Recall, $\{\bfu_j\}$ are the
eigenvectors of $\bfch$ and $\{\bfv_j\}$ are the normalized adjoint vectors.

With  $\bfy^{\pm k} = (\bfp_{\pm k} , \bfq_{\pm k}, 0)^T$, from Lemma \ref{lem:T}, we know
the ladder operators can be written as
\begin{equation*}
  \mcl_k = \bfp_k \cdot \nabla + \bfq_k \cdot \bs, \quad k=-n,\ldots,-1,1,\ldots,n,
\end{equation*}
with $\bfp_{\pm k}$ and $\bfq_{\pm k}$ given explicitly in the proof of Lemma \ref{lem:T}.
From these we see that for \eq{as} to be satisfied we must have
\begin{equation}\label{eq:coeffs2}
    \sum_{k=1}^n \alpha_k \ob{\bfv}_k = \bfa,\quad 
\sum_{k=1}^n \beta_k \bfu_k = \sum_{j=1}^n \alpha_j (\bfch -\mu_j \bfci)^{-1}\bfcb \ob{\bfv}_j.
\end{equation}
Hence, $\ob{\bfcv}{\boldsymbol \alpha} = \bfa$, where $\bfcv=[\bfv_1 , \bfv_2, \ldots , \bfv_n]$.
 But $\bfcv^*\bfcu = \bfci$, so the first expression in
\eq{coeffs2} is equivalent to the definition of $\alpha_k$ in \eq{coeffs1}. 
Also, since the $\{ \bfu_k\}$ are complete,
and $((\bfch-\mu_j \bfci)^{-1})^*\bfv_k =-(\ob{\mu}_k + \ob{\mu}_j)^{-1} \bfv_k$
we conclude
\begin{align}
\beta_k & =   \ip{\bfv_k}{\sum_{j=1}^n \alpha_j 
(\bfch -\mu_j \bfci)^{-1}\bfcb \ob{\bfv}_j} \nonumber \\
& =  -\sum_{j=1}^n \frac{\alpha_j}{\mu_k + \mu_j}   \ip{\bfv_k}{\bfcb \ob{\bfv}_j} 
= -\sum_{j=1}^n \frac{\ob{\alpha}_j}{\mu_k + \ob{\mu}_j}   \ip{\bfv_k}{\bfcb \bfv_j}, \nonumber
\end{align}
where the last equality follows from a rearrangement of the sum over $j$, and the
fact that the eigenvectors $\bfv_j$ and eigenvalues $\mu_j$ come in conjugate pairs.
Thus, with $\alpha_k,\, \beta_k$ defined as in \eq{coeffs1}, the equations in \eq{coeffs2}
are satisfied, and therefore \eq{as} holds.
\end{proof}



\begin{center}
{\bf Proof of Theorem \ref{thm:m2} }
\end{center}

\begin{proof}
From Lemma \ref{lem:quick} we 
have 
\begin{align}
  \lambda_2 &=-\ip{\bfpsi_1}{\sum_{k=1}^n\alpha_k \gone \bfc_k}=
 -\ip{\bfpsi_1}{ \sum_{m=1}^n \sum_{j=1}^J \frac{\alpha_m \beta_m \ip{\bfpsi_j}{\gone \bfphi_1}}
{\nu_1 - \nu_j + \mu_m}\gone \bfphi_j } \nonumber \\
& =  -\sum_{j=1}^J \ip{\bfpsi_1}{\gone \bfphi_j}\ip{\bfpsi_j}{\gone \bfphi_1}
\sum_{m=1}^n \frac{\alpha_m \beta_m }
{\nu_1 - \nu_j + \mu_m} \nonumber \\
& = \sum_{j=1}^J \ip{\bfpsi_1}{\gone \bfphi_j}\ip{\bfpsi_j}{\gone \bfphi_1}
G(\nu_1 - \nu_j).\nonumber 
\end{align}
The last equality follows from \eq{G}.
\end{proof}


\nocite{Arnold} \nocite{Dirac} \nocite{Kloeden} \nocite{Asmussen} 
\nocite{Lamb} \nocite{Kampen}
\bibliographystyle{plain}
\bibliography{ladder}

\def\cprime{$'$}
\begin{thebibliography}{10}

\bibitem{Marsden}
Ralph Abraham and Jerrold~E. Marsden.
\newblock {\em Foundations of mechanics}.
\newblock Benjamin/Cummings Publishing Co. Inc. Advanced Book Program, Reading,
  Mass., 1978.
\newblock Second edition, revised and enlarged, With the assistance of Tudor
  Ra{\c{t}}iu and Richard Cushman.

\bibitem{adams-bloch2008}
Fred~C. Adams and Anthony~M. Bloch.
\newblock Hill's equation with random forcing terms.
\newblock {\em SIAM J. Appl. Math.}, 68(4):947--980, 2008.

\bibitem{adams-bloch2009}
Fred~C. Adams and Anthony~M. Bloch.
\newblock Hill's equation with random forcing parameters: the limit of delta
  function barriers.
\newblock {\em J. Math. Phys.}, 50(7):073501, 20, 2009.

\bibitem{adams-bloch2010}
Fred~C. Adams and Anthony~M. Bloch.
\newblock Hill's equation with random forcing parameters: determination of
  growth rates through random matrices.
\newblock {\em J. Stat. Phys.}, 139(1):139--158, 2010.

\bibitem{Arnold}
Ludwig Arnold.
\newblock Stochastic differential equations as dynamical systems.
\newblock In {\em Realization and modelling in system theory ({A}msterdam,
  1989)}, volume~3 of {\em Progr. Systems Control Theory}, pages 489--495.
  Birkh\"auser Boston, Boston, MA, 1990.

\bibitem{ArnoldV}
V.~I. Arnol{\cprime}d.
\newblock {\em Geometrical methods in the theory of ordinary differential
  equations}, volume 250 of {\em Grundlehren der Mathematischen Wissenschaften
  [Fundamental Principles of Mathematical Sciences]}.
\newblock Springer-Verlag, New York, second edition, 1988.
\newblock Translated from the Russian by Joseph Sz{\"u}cs [J{\'o}zsef M.
  Sz{\H{u}}cs].

\bibitem{Asmussen}
S{\o}ren Asmussen and Peter~W. Glynn.
\newblock {\em Stochastic simulation: algorithms and analysis}, volume~57 of
  {\em Stochastic Modelling and Applied Probability}.
\newblock Springer, New York, 2007.

\bibitem{SAND}
Timothy Blass and L.A. Romero.
\newblock On the stability of stochastically forced parametric oscillators.
\newblock {\em Sandia National Laboratories Report}, SAND-2012-6980, 2012.

\bibitem{bobryk-chrz}
R.~V. Bobryk and A.~Chrzeszczyk.
\newblock Colored-noise-induced parametric resonance.
\newblock {\em Physica A}, 316:225--232, 2002.

\bibitem{carlson}
David Carlson, B.~N. Datta, and Hans Schneider.
\newblock On the controllability of matrix pairs {$(A,\,K)$} with {$K$}
  positive semidefinite.
\newblock {\em SIAM J. Algebraic Discrete Methods}, 5(3):346--350, 1984.

\bibitem{Dirac}
P.~A.~M. Dirac.
\newblock {\em The {P}rinciples of {Q}uantum {M}echanics}.
\newblock Oxford, at the Clarendon Press, 1947.
\newblock 3d ed.

\bibitem{Gardiner}
C.~W. Gardiner.
\newblock {\em Handbook of stochastic methods}, volume~13 of {\em Springer
  Series in Synergetics}.
\newblock Springer-Verlag, Berlin, second edition, 1985.
\newblock For physics, chemistry and the natural sciences.

\bibitem{khas}
R.Z. KhasÊ¹minski.
\newblock {\em {Stochastic stability of differential equations}}.
\newblock Kluwer Academic Pub, 1980.

\bibitem{Kloeden}
Peter~E. Kloeden and Eckhard Platen.
\newblock {\em Numerical solution of stochastic differential equations},
  volume~23 of {\em Applications of Mathematics (New York)}.
\newblock Springer-Verlag, Berlin, 1992.

\bibitem{Lamb}
Horace Lamb.
\newblock {\em Hydrodynamics}.
\newblock Cambridge Mathematical Library. Cambridge University Press,
  Cambridge, sixth edition, 1993.
\newblock With a foreword by R. A. Caflisch [Russel E. Caflisch].

\bibitem{lee-etal}
Hye~Jin Lee, Changho Kim, Jae~Gil Kim, and Eok~Kyun Lee.
\newblock A general scheme for studying the stochastic dynamics of a parametric
  oscillator driven by coloured noise.
\newblock {\em J. Phys. A}, 37(3):647--656, 2004.

\bibitem{Liberzon}
Daniel Liberzon and Roger~W. Brockett.
\newblock Spectral analysis of {F}okker-{P}lanck and related operators arising
  from linear stochastic differential equations.
\newblock {\em SIAM J. Control Optim.}, 38(5):1453--1467, 2000.

\bibitem{Metafune}
G.~Metafune, D.~Pallara, and E.~Priola.
\newblock Spectrum of {O}rnstein-{U}hlenbeck operators in {$L^p$} spaces with
  respect to invariant measures.
\newblock {\em J. Funct. Anal.}, 196(1):40--60, 2002.

\bibitem{repetto}
R.~Repetto and V.~Galletta.
\newblock {Finite amplitude Faraday waves induced by a random forcing}.
\newblock {\em Physics of fluids}, 14:4284, 2002.

\bibitem{resibois}
P.M.V. R{\'e}sibois and M.~De~Leener.
\newblock {\em {Classical kinetic theory of fluids}}.
\newblock Wiley New York, 1977.

\bibitem{risken}
H.~Risken.
\newblock {\em The {F}okker-{P}lanck equation}, volume~18 of {\em Springer
  Series in Synergetics}.
\newblock Springer-Verlag, Berlin, second edition, 1989.
\newblock Methods of solution and applications.

\bibitem{roy}
R.~Val{\'e}ry Roy.
\newblock Stochastic averaging of oscillators excited by colored {G}aussian
  processes.
\newblock {\em Internat. J. Non-Linear Mech.}, 29(4):463--475, 1994.

\bibitem{titulaer}
U.~M. Titulaer.
\newblock A systematic solution procedure for the {F}okker-{P}lanck equation of
  a {B}rownian particle in the high-friction case.
\newblock {\em Phys. A}, 91(3-4):321--344, 1978.

\bibitem{Kampen}
N.~G. van Kampen.
\newblock {\em Stochastic processes in physics and chemistry}, volume 888 of
  {\em Lecture Notes in Mathematics}.
\newblock North-Holland Publishing Co., Amsterdam, 1981.

\bibitem{wilkinson}
Michael Wilkinson.
\newblock Perturbation theory for a stochastic process with
  {O}rnstein-{U}hlenbeck noise.
\newblock {\em J. Stat. Phys.}, 139(2):345--353, 2010.

\end{thebibliography}

\end{document}